\documentclass[preprint,1p,times]{elsarticle}
\usepackage{amssymb}
\usepackage{amsmath}
\usepackage{amsthm}
\theoremstyle{plain}
\newtheorem{theorem}{Theorem}[section]
\newtheorem{lemma}[theorem]{Lemma}
\newtheorem{proposition}[theorem]{Proposition}
\newtheorem{corollary}[theorem]{Corollary}
\theoremstyle{definition}

\theoremstyle{remark}
\newtheorem{remark}[theorem]{Remark}

\usepackage{natbib}
\usepackage{mathtools}
\usepackage{bm}
\usepackage{blkarray}
\usepackage{nomencl}
\usepackage{multicol}
\usepackage[most]{tcolorbox}
\usepackage{etoolbox}
\usepackage{lineno}
\usepackage{url}

\begin{document}

\begin{frontmatter}

\title{Non-Negativity, Non-Singularity and Energy Conservation of the Zonal Method with Non-Reciprocal Exchange Factors}

\author{Nikolaj Maack Bielefeld\corref{cor1}}

\ead{nikolajbielefeld@gmail.com}
\ead[url]{https://gert.net}
\cortext[cor1]{Corresponding author}
\affiliation{organization={Independent researcher},
            addressline={Karl Bjarnhofs Vej 1A, st 2}, 
            postcode={7120}, 
            city={Vejle East},
            state={Southern Jutland},
            country={Denmark}}

\begin{abstract}
This paper proves that the zonal balances of Hottel's method, written in power variables for a first-interaction exchange factor matrix $\mathbf{F}$, have a unique and non-negative solution for any combination of prescribed temperatures and prescribed source terms in which at least one element with prescribed temperature absorbs radiation, and that they conserve energy as an algebraic identity. The results hold for media ranging from transparent to absorbing, emitting and scattering, and for boundaries ranging from absorbing to reflecting. They require only that $\mathbf{F}$ be non-negative, row stochastic and irreducible: reciprocity is not needed, so the guarantees hold for Monte Carlo exchange factors as traced, without smoothing, which the symmetric notation of the zonal method does not make apparent. Both the non-singularity and the non-negativity follow from a single discrete minimum principle, and both have been machine-checked in Lean 4. To this end the balances are written as one linear system in the total radiant powers, $\mathbf{M}\mathbf{j} = \mathbf{h}$, assembled row by row from the prescribed temperatures and source terms, with a system matrix that inherits the sparsity of $\mathbf{F}$; the total exchange areas of Noble's matrix formulation are in general dense, and the two agree symbolically and numerically to machine precision. Validation is provided symbolically against the textbook closed-form solutions for infinite parallel plates and concentric cylinders, numerically against the results of Crosbie and Schrenker for pure and partial scattering, against an exact symmetry identity in a three-dimensional enclosure, and against the diffusion approximation in an optically thick problem with about $10^6$ elements, solved in about one minute.
\end{abstract}

\begin{graphicalabstract}
\includegraphics[width=1.0\linewidth]{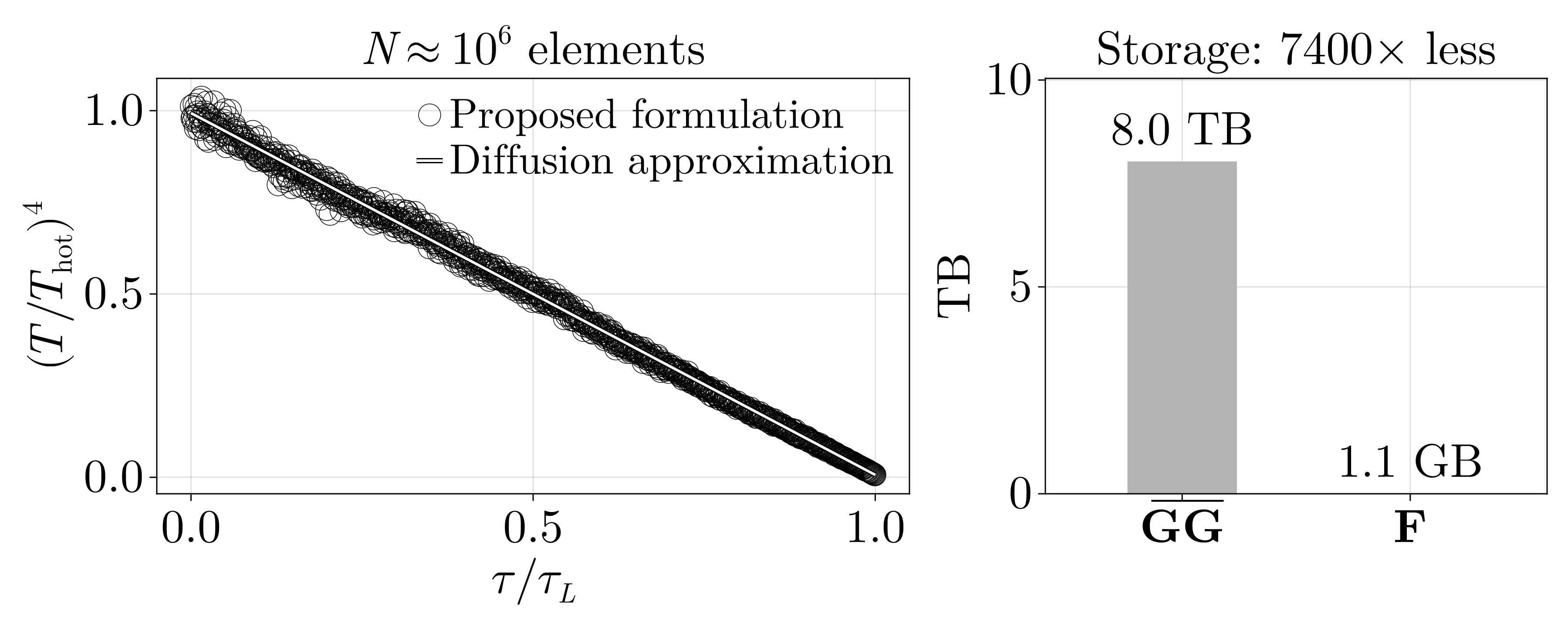}
\end{graphicalabstract}

\begin{highlights}
\item Mixed-boundary system proven non-singular, with the proof machine-checked in Lean 4
\item Guaranteed non-negative radiation to machine precision for non-negative sources
\item Unconditional energy conservation to machine precision as an algebraic identity
\item Guarantees hold for non-reciprocal, unsmoothed Monte Carlo exchange factors
\item Proven to be Hottel's zonal balances in power variables, for any exchange factors
\end{highlights}

\begin{keyword}
RTE
\sep
Radiative transfer
\sep
Exchange factors
\sep
Zonal method
\sep
Multiple scattering
\sep
Energy conservation
\sep
Matrix methods
\end{keyword}

\end{frontmatter}

\section{Introduction}

\subsection{The Radiative Transfer Equation}

Radiative transfer in participating media is a fundamental process in many natural and engineered systems, from stellar atmospheres to thermal engineering. The theoretical foundation for radiative transfer was established in the 19th century by Gustav Kirchhoff, who formulated the fundamental relationship between emission and absorption of thermal radiation and introduced the concept of blackbody radiation \citep{Kirchhoff1860}. Building upon this work, Karl Schwarzschild developed early formulations for radiative transfer in stellar atmospheres in the early 1900s \citep{Schwarzschild1906}. The modern mathematical framework of the radiative transfer equation (RTE) was systematically developed by Subrahmanyan Chandrasekhar in his seminal 1950 treatise "Radiative Transfer" \citep{Chandrasekhar1950}, which provided rigorous analytical methods and established the equation as the cornerstone of radiation transport theory across diverse fields from astrophysics to engineering applications. Despite this rich theoretical heritage, the fundamental governing equation remains a challenging integro-differential equation in seven independent variables: Three spatial coordinates ($x$, $y$, $z$), two directional angles ($\theta$, $\phi$), one spectral wavelength ($\lambda$) and one temporal time ($t$). Under steady-state and grey media assumptions, the temporal ($t$) and spectral ($\lambda$) dependencies are eliminated, reducing the RTE from seven to five independent variables, which yields the differential form shown in equation \eqref{eq:RTE} \citep{Daun2021}:

\begin{equation}
\begin{split}
\frac{\partial i(S,\Omega)}{\partial S} = \; &
\kappa i_\mathrm{b}(S) 
- \kappa i(S,\Omega)
- \sigma_\mathrm{s} i(S,\Omega) \\
& + \frac{\sigma_\mathrm{s}}{4 \pi} \int_{\Omega_i=4 \pi}
i(S,\Omega_i) \Phi(\Omega_i,\Omega) d\Omega_i
\end{split}
\label{eq:RTE}
\end{equation}
where $i(S,\Omega)$ is the intensity along a ray trajectory, as a function of position $S$ and direction $\Omega\equiv\Omega(\theta,\phi)$, $\kappa$ is the absorption coefficient of the medium, $i_\mathrm{b}(S)$ is the blackbody intensity, $\sigma_\mathrm{s}$ is the scattering coefficient and $\Phi$ is the scattering phase function and the integral is over all solid angles (4$\pi$ steradians).

The RTE of equation \eqref{eq:RTE} is a differential energy balance along the trajectory $S$ of a single ray. To solve it one could, in principle, simply integrate along the ray trajectory. What makes the RTE particularly challenging to solve for general geometries is that all possible ray trajectories in the domain are coupled through two mechanisms: (1) The state of the participating medium, i.e. the temperature of the intervening gas, which depends on the local radiative heating rate and governs both the rate of emission (through blackbody emission laws) and temperature-dependent absorption properties, and (2) the in-scattering term, which requires integration over all incoming directions to account for radiation scattered from every direction into the direction of $S$. Therefore, for general geometries, to determine the equilibrium state, all possible ray trajectories must be solved simultaneously. The RTE does not directly account for surface reflections, requiring radiative transfer formulations that couple the RTE with separate treatment of boundary reflection-scattering.

\subsection{Zonal, Exchange Factor and Monte Carlo Methods}

Seeking to quantify combustion heat transfer, Hottel and co-workers \citep{Hottel1958, Hottel1967} pioneered solutions to the RTE using a discretized integral approach. Fundamentally, this integral approach produces \textit{exchange areas}, which describe the radiative connectivity of the discretized domain. The integrals for the exchange areas are:\\
Surface-surface:
\begin{equation}
	\overline{s_j s_k} = \frac{1}{\pi}\int_{A_j}\int_{A_k} \frac{\overline{t}(S_{j-k})\mathrm{cos}(\theta_j)\mathrm{cos}(\theta_k)}{S^2_{j-k}}dA_k dA_j
	\label{eq:viewfactor_surfsurf}
\end{equation}
Surface-gas and gas-surface:
\begin{equation}
	\overline{s_k g_\gamma} = \overline{g_\gamma s_k} = \frac{\kappa}{\pi} \int_{V_\gamma}\int_{A_k}\frac{\overline{t}(S_{k-\gamma}) \mathrm{cos}(\theta_k)}{S^2_{k-\gamma}} dA_k dV_\gamma
	\label{eq:viewfactor_gassurf}
\end{equation}
Gas-gas:
\begin{equation}
	\overline{g_\gamma g_{\gamma*}} = \frac{\kappa^2}{\pi} \int_{V_\gamma}\int_{V_{\gamma*}} \frac{\overline{t}(S_{\gamma-\gamma*})}{S^2_{\gamma-\gamma*}}dV_{\gamma*} dV_\gamma
	\label{eq:viewfactor_gasgas}
\end{equation}
where $\overline{t}=\mathrm{exp}(-\kappa S)$ is the transmissivity of the participating medium, $\theta$ is the angle with the surface normal, $\mathrm{cos}(\theta)$ of the emitter surface accounts for the diffuse Lambertian emission and $\mathrm{cos}(\theta)$ of the absorbing surface accounts for the projected area effect, $S$ is the distance from the point of emission to the point of absorption and $\kappa$ is the absorption coefficient of the medium. These integrals, being analytically intractable, were initially, before the general advent of modern-day digital computers, solved graphically and tabulated. Once the exchange areas have been determined for all pairs of emitters and absorbers in the system, subsequent solution of the RTE is achieved using zonal balances.

Noble \citep{Noble1975} cast the zonal balances in matrix form, including diffuse reflection at the walls and isotropic scattering in the medium, by eliminating the radiosities to obtain explicit matrix relations for the total exchange areas, which give the power exchanged between every pair of zones over all paths. Naraghi and Chung \citep{Naraghi1985} rederived these relations by treating radiative exchange as a Markov chain, whose transition probabilities are the direct exchange areas divided by the emission capacity of the emitting zone, and combined them into a single matrix relation. Goyhénèche and Sacadura \citep{Goyheneche2002} extended the explicit relations to anisotropic scattering and anisotropically reflecting walls. Common to these formulations is that the total exchange areas couple every pair of zones connected by multiple reflection or scattering, so that they are dense matrices even when the direct exchange areas are sparse.

The transition probabilities of Naraghi and Chung are dimensionless exchange factors, and such factors have a longer history in the treatment of radiative exchange. Gebhart \citep{Gebhart1961} introduced absorption factors for enclosures of grey, diffuse surfaces, giving the fraction of the energy emitted by one surface that is absorbed by another, including all reflections. Larsen \citep{Larsen1983} and Larsen and Howell \citep{LarsenHowell1985} extended the idea to participating media in the exchange factor method. Their exchange factors account for all paths between two zones through absorption, re-emission and isotropic scattering in the medium, but exclude wall reflections, which are resolved through the energy balances as in the zonal method. Since isotropic scattering and absorption followed by re-emission redistribute energy alike, these exchange factors depend only on the geometry and the total extinction, not on the single scattering albedo, a property Larsen proposed to exploit by measuring exchange factors in scale models filled with a scattering medium.

With the first digital computers, Metropolis, Ulam and co-workers pioneered the Monte Carlo method \citep{MetropolisUlam1949}, a statistical method that can model complex phenomena which are analytically intractable. Howell and Perlmutter \citep{HowellPerlmutter1964} applied it to radiative transfer through an absorbing and emitting grey gas between grey walls, and it has since become a standard tool in thermal radiation \citep{MonteCarloPastFuture2021}, where it evaluates multi-dimensional integrals such as equations \eqref{eq:viewfactor_surfsurf}--\eqref{eq:viewfactor_gasgas} by tracing rays, and is a common means of obtaining exchange factors in participating media, where these integrals have no closed form. The statistical error of such estimates can be reduced by sampling the emission points and directions of the rays from low-discrepancy sequences, such as the Sobol sequences \citep{Sobol1967, JoeKuo2008}, in place of pseudorandom numbers. Other approaches solve the RTE directly, by discretizing the angular dependence of the intensity, as in the discrete ordinates and finite volume methods, or by expanding it in spherical harmonics. Comprehensive reviews of these methods can be found in \citep{Daun2021} and \citep{Modest2022}.

\subsection{The Present Formulation}

The formulation of the present work is Hottel's zonal method, outlined above, written in different variables. In place of the exchange area matrix it uses the non-dimensional exchange factor matrix, in which each row of the exchange areas has been normalized by the emission capacity of its emitter, so that every row sums to unity: the radiation leaving any element is fully accounted for across its destinations. In place of radiosities it uses the total radiant powers of the elements, in Watts. The exchange factor matrix is partitioned at the single-step level into the fraction of incident radiation absorbed and the fraction reflected or scattered, and the zonal balances are assembled into one linear system for any combination of prescribed temperatures and prescribed source terms. The multi-bounce structure of the exchange arises algebraically from the solution of this system, whose matrix inherits the sparsity of the exchange factor matrix. Section~\ref{sec:hottel_equivalence} proves that this system and Hottel's zonal balances are the same discrete model for every exchange factor matrix, reciprocal or not. What the variables add is that the properties of the system become visible in its matrix: the proofs of non-singularity, non-negativity and unconditional energy conservation in Sections~\ref{sec:matrix_props} and \ref{sec:physical_correct} are short, require only that the rows of the exchange factor matrix sum to unity and that the domain be connected, and carry over to Hottel's balances.

As in the zonal method, the accuracy of the proposed formulation is limited by the zoning and by the exchange factor matrix $\mathbf{F}$, and the formulation introduces no further approximation. For convex enclosures of planar polygons with a transparent medium, $\mathbf{F}$ can be computed analytically with the method of Narayanaswamy \citep{Narayanaswamy2015}, which is used for the three-dimensional test of Section~\ref{sec:3d}. For participating media no general closed form is available, and $\mathbf{F}$ is commonly estimated by Monte Carlo ray tracing, whose statistical noise violates reciprocity. Smoothing methods restore reciprocity and closure, from the least-squares smoothing of Larsen and Howell \citep{LarsenHowell1986} to the constrained maximum likelihood estimation of Daun et al. \citep{Daun2005}. In this paper, the exchange factors are used as traced, except in the comparison with Noble's explicit relations, which take the direct exchange areas as a symmetric matrix.

The proposed formulation is discrete throughout: it operates on zones and on the exchange factors between them, and the continuous physics of ray propagation enters only through $\mathbf{F}$. Exchange factor methods thereby separate the radiative transfer problem into two parts: the propagation of rays through the geometry and the participating medium, which determines $\mathbf{F}$, and the energy balance between the zones, which the proposed formulation solves. Since only the point of emission and the point of first interaction of each ray enter $\mathbf{F}$, any propagation model within geometric optics can in principle supply it without changing the formulation. In a medium with a continuously varying refractive index, for example, rays follow curved paths given by the ray equation \citep{BornWolfOptics1980}, although the emission capacities and the reciprocity relation then also carry the square of the refractive index. For compactness, the present paper treats the discrete formulation exclusively, with exchange factors for straight rays in media of uniform refractive index.

For the proposed formulation, since $\mathbf{F}$ refers only to the first interaction, Monte Carlo methods used to obtain $\mathbf{F}$ can be truncated after this initial interaction, making them computationally less demanding than basic Monte Carlo implementations that must track multiple reflection-scattering events \citep{MonteCarloPastFuture2021}. This will be increasingly true for higher single reflection-scattering coefficients due to the increased length of the ray paths. Instead, the proposed formulation performs the multiple reflection-scattering ray tracing analytically using solution of a linear system.

\subsection{Contributions and Outline}

The main contributions of this paper are:
\begin{enumerate}
\item Proven non-singularity of the zonal balances for any combination of prescribed temperatures and prescribed source terms in which at least one element with prescribed temperature absorbs radiation (Theorem~\ref{thm:M_nonsingular}).
\item Proven non-negative radiation for non-negative source terms (Theorem~\ref{thm:physical_correctness}).
\item Unconditional energy conservation as an algebraic identity (Theorem~\ref{thm:unconditional_conservation}).
\item Validity of items 1--3 for any non-negative, row-stochastic and irreducible $\mathbf{F}$, without reciprocity, so that ray-traced exchange factors retain them without smoothing (Remark~\ref{rem:nonreciprocal}); items 1 and 2 follow from a single discrete minimum principle (Lemma~\ref{lem:min_principle}) and have been machine-checked.
\item A formulation of the zonal balances in power variables as one linear system whose matrix inherits the sparsity of $\mathbf{F}$, proven to coincide with Hottel's balances for every $\mathbf{F}$ (Theorem~\ref{thm:hottel_equivalence}), from which Noble's total exchange areas follow in closed form (Proposition~\ref{prop:total_exchange_areas}); the system is solved at about $10^6$ elements, where the total exchange areas are not formed.
\end{enumerate}

The remainder of the paper is structured as follows: Section \ref{sec:methodology1} presents the methodology. Section \ref{sec:matrix_props} establishes the matrix properties of the formulation, including a discrete minimum principle from which both the non-singularity of the mixed-boundary system and the non-negativity of its solution follow. Section \ref{sec:physical_correct} establishes physical correctness through guaranteed non-negative radiation, unconditional energy conservation and the invariance of radiative equilibrium solutions under the single scattering albedo. Section \ref{sec:hottel_equivalence} proves the equivalence of the formulation with Hottel's zonal balances and recovers Noble's total exchange areas from it. Section \ref{sec:validation} validates the formulation symbolically, numerically and by machine-checked proofs in Lean 4. Section \ref{sec:discussion_overall} discusses scalability, computational complexity, comparison to existing methods, limitations, and future work. Finally, Section \ref{sec:conclusion} presents the conclusion.

\section{Methodology}
\label{sec:methodology1}

The methodology is split into two parts: First, the exchange factor matrix is formally defined, and next, the exchange factor formulation is applied to solution of the RTE.

\subsection{Definition of the Exchange Factor Matrix}

The proposed formulation requires as an input a general description of the radiative connectivity of the domain in the form of the exchange factor matrix $\mathbf{F}$. The exchange factor matrix $\mathbf{F}$ may be obtained analytically or numerically. The $(m+n)\times(m+n)$ row-stochastic blocked first interaction exchange factor matrix $\mathbf{F}$ for a system of $m$ bounding surfaces and $n$ gas volume elements is defined as:
\begin{equation}
	\renewcommand{\arraystretch}{1.4}
	\mathbf{F} =
	\begin{blockarray}{ccl}
	\scriptstyle m & \scriptstyle n & \\
	\begin{block}{[cc]l}
	\mathbf{F}_\mathrm{ss} & \mathbf{F}_\mathrm{sg} & \scriptstyle m \\
	\mathbf{F}_\mathrm{gs} & \mathbf{F}_\mathrm{gg} & \scriptstyle n \\
	\end{block}
	\end{blockarray}
\end{equation}
where $\mathbf{F}$ is row stochastic to ensure that conservation of energy is satisfied and the entries of $\mathbf{F}$ are the fractions of energy emitted by the emitter of row $i$ which has its first interaction with the element of column $j$. The entries of $\mathbf{F}$ can be defined from equations \eqref{eq:viewfactor_surfsurf}-\eqref{eq:viewfactor_gasgas} by normalizing with the equivalent total capacity for emission and reflection-scattering of the emitter zone (and using the extinction coefficient $\beta$ in place of the absorption coefficient $\kappa$):
\begin{equation}
	F_{s_j s_k} = \frac{\overline{s_j s_k}}{A_j}, \quad
	F_{s_k g_\gamma} = \frac{\overline{s_k g_\gamma}}{A_k}, \quad
	F_{g_\gamma s_k} = \frac{\overline{g_\gamma s_k}}{4\beta V_\gamma}, \quad
	F_{g_\gamma g_{\gamma*}} = \frac{\overline{g_\gamma g_{\gamma*}}}{4\beta V_\gamma}.
\end{equation}
For an opaque surface, the reflectivity $\rho$ and the emissivity $\varepsilon$ sum to unity, so its capacity to emit and reflect is its area $A$. For a volume of participating medium, the corresponding capacity is $4\beta V$, with the extinction coefficient $\beta = \kappa + \sigma_\mathrm{s}$, where $\sigma_\mathrm{s}$ is the scattering coefficient, so that it includes both emission and scattering. This normalization ensures that $\mathbf{F}$ accounts for all outgoing radiation, whether emitted or reflected-scattered.

The exchange factors used here are therefore direct-exchange areas normalized by the emission and reflection-scattering capacity of the emitter; they differ from the exchange factors of Larsen and Howell \citep{LarsenHowell1985}, which also include paths through absorption, re-emission and scattering in the medium, and to which $\mathbf{F}$ reduces only in the optically thin limit. The exact exchange factors additionally satisfy reciprocity \citep{Daun2005}, which ray-traced estimates satisfy only up to their statistical noise:
\begin{equation}
 \mathbf{W} \mathbf{F} = \mathbf{F}^\top  \mathbf{W}
 \label{eq:reciprocity}
\end{equation}
where $\mathbf{W}=\mathrm{diag}(\mathbf{w})$ is a diagonal matrix of equivalent emission areas with $w_i = A_i$ for $i\leq m$ and $w_i=4 \beta_i V_i$ for $i>m$, the full capacity of each element to emit and reflect or scatter.

Furthermore, in the general case, $\mathbf{F}$ will have some zero entries, i.e., physical elements without direct view of each other, meaning $\mathbf{F}$ is not merely a positive matrix but a non-negative matrix. Therefore, for the non-singularity proof of Theorem~\ref{thm:M_nonsingular}, $\mathbf{F}$ must be \textit{irreducible} (strongly connected) \citep{HornJohnson2018}, where irreducibility corresponds to the physical requirement that the domain contains no isolated regions, that is, for any two elements $i$ and $j$, there exists a sequence of direct radiation exchanges $i \to k_1 \to k_2 \to \ldots \to j$ connecting them. This requirement can always be satisfied by treating isolated regions as separate radiative transfer problems.

For the present formulation, $\mathbf{F}$ should be derived exactly as if dealing with pure absorption, but always using the total capacity to emit and reflect-scatter instead of the absorption coefficient. Then after $\mathbf{F}$ has been obtained, any part of the capacity to absorb may be converted into a capacity for reflection-scattering elementwise: A single $\mathbf{F}$ therefore covers all combinations of $\kappa_i+\sigma_{\mathrm{s},i}$ which sum to $\beta_i$.

\subsection{The Exchange Factor Formulation}

The column constant $(m+n)\times(m+n)$ blocked single reflection-scattering matrix $\mathbf{B}$ is defined as:
\begin{equation}
	\renewcommand{\arraystretch}{1.4}
	\mathbf{B} =
	\begin{blockarray}{ccl}
	\scriptstyle m & \scriptstyle n & \\
	\begin{block}{[cc]l}
	\mathbf{B}_\mathrm{ss} & \mathbf{B}_\mathrm{sg} & \scriptstyle m \\
	\mathbf{B}_\mathrm{gs} & \mathbf{B}_\mathrm{gg} & \scriptstyle n \\
	\end{block}
	\end{blockarray}
	=
	\begin{bmatrix}
	\rho_1 \quad \dots \quad \rho_m \quad \omega_1  \quad \dots \quad \omega_n \\
	\vdots \quad \ddots \quad\quad \vdots \quad\quad \vdots \quad \ddots \quad \vdots \\
	\rho_1 \quad \dots  \quad \rho_m \quad \omega_1  \quad \dots  \quad \omega_n
	\end{bmatrix}
\end{equation}
where $\rho_i$ is the reflectivity of boundary element $i$ and $\omega_i=\sigma_{\mathrm{s},i}/\beta_i$ is the single scattering albedo of volume element $i$. 

The absorption matrix $\mathbf{A}$ and the reflection-scattering matrix $\mathbf{R}$ are expressed as:
\begin{equation}
	\mathbf{A} = \mathbf{F} \circ (\mathbf{1} - \mathbf{B})\label{eq:A}
\end{equation}
\begin{equation}
	\mathbf{R} = \mathbf{F} \circ \mathbf{B}
	\label{eq:R}
\end{equation}
Here $\circ$ denotes the Hadamard (entry-wise) product and $\mathbf{1}$ the $(m+n)\times(m+n)$ matrix of ones; elsewhere, $\mathbf{1}$ denotes the vector of ones. The first matrix, $\mathbf{F}$, accounts for the probability of transmission, and the last matrices, $\mathbf{1} - \mathbf{B}$ or $\mathbf{B}$, account for absorption or reflection at the destination.

Next, based on fundamental radiative energy balances, which state that an energy source is equal to the total radiant power, minus the absorbed power, minus the reflected incident power, and that the emissive power equals the total radiant power, minus the reflected incident power, the matrices $\mathbf{C}$ and $\mathbf{D}$ are defined as:
\begin{equation}
	\mathbf{C} = \mathbf{I} - \mathbf{A}^\top - \mathbf{R}^\top
	\quad , \quad
	\mathbf{D} = \mathbf{I} - \mathbf{R}^\top
	\label{eq:C_and_D}
\end{equation}

Now, given a vector $\mathbf{e}$ of the emissive powers and a vector $\mathbf{q}$ of the source terms of each element, both in the units of Watts, where only one entry of these can be specified for each element, the mixed-boundary matrix $\mathbf{M}$ can be assembled row by row:
\begin{equation}
	\mathbf{M}[i,:] = \mathbf{C}[i,:] \quad , \quad \mathbf{M}[j,:] = \mathbf{D}[j,:]
\end{equation}
using $i$ for known entries of the source vector $q_i$ and using $j$ for known entries of the emissive power vector $e_j$, where $i$ and $j$ are mutually exclusive. Next, the linear system:
\begin{equation}
	\mathbf{M}\mathbf{j} = \mathbf{h}
	\label{eq:solve_for_j}
\end{equation}
is solved for the total radiant power $\mathbf{j}$ of each element, where $\mathbf{h}$ is a vector of known emissive powers or source terms, also in the units of Watts, depending on which is known, corresponding to the $i$ and $j$ rows of $\mathbf{M}$.

Knowing the total radiant power of each element, the absorbed incident power $\mathbf{g}_\mathrm{a}$ and the unknown terms of the emissive power vector $\mathbf{e}$ and the source term vector $\mathbf{q}$ can be determined from:
\begin{equation}
	\mathbf{g}_\mathrm{a} = \mathbf{A}^\top \mathbf{j} 
	\quad , \quad
	e_i = q_i + g_{\mathrm{a},i}
	\quad, \quad 
	q_j = e_j - g_{\mathrm{a},j}
	\label{eq:emissive_source}
\end{equation}
where $i$ and $j$ are again mutually exclusive. The reflected-scattered power $\mathbf{r}$ and the total incident power $\mathbf{g}$ can be determined from:
\begin{equation}
	\mathbf{r} = \mathbf{R}^\top \mathbf{j} \quad , \quad
	\mathbf{g} = \mathbf{g}_\mathrm{a} + \mathbf{r}
	\label{eq:r_and_g}
\end{equation}
The total intensity of gas volumes and of Lambertian surfaces can be determined from:
\begin{equation}
	i_{\mathrm{g},i} = \frac{j_{\mathrm{g},i}}{4\pi  \beta_i V_i} 
	\quad , \quad
	i_{\mathrm{s},i} = \frac{j_{\mathrm{s},i}}{\pi A_i}
\end{equation}
Finally the temperature of each element can be determined from its emissive power using $\varepsilon$ and $\kappa$ and the Stefan-Boltzmann law:
\begin{equation}
	T_{\mathrm{s},i} = \left( \frac{e_{\mathrm{s},i} }{ \varepsilon_{\mathrm{s},i}\sigma A_{\mathrm{s},i} }\right)^{1/4}
	\quad , \quad
	T_{\mathrm{g},i} = \left( \frac{e_{\mathrm{g},i} }{4 \kappa_i\sigma V_{\mathrm{g},i} n_i^2 }\right)^{1/4}	
\label{eq:temperature}
\end{equation}
where $\sigma$ is the Stefan-Boltzmann constant and $n_i$ is the refractive index of the medium of element $i$.

\section{Matrix Properties}
\label{sec:matrix_props}

This section will explore the properties of the matrices of the proposed exchange factor formulation. These properties will subsequently be used to establish conditions which guarantee physical correctness.

\subsection{Spectral Radii}

This section determines the spectral radius of $\mathbf{A}+\mathbf{R}$, an upper bound on the spectral radius of $\mathbf{R}$, and, as a consequence of the latter, a bound on the condition number of $\mathbf{D}$.

\begin{theorem}[Spectral Radius of Combined Absorption and Reflection-Scattering]
\label{thm:spectral_radius_unity}
Let $\mathbf{F}$ be the row stochastic and irreducible exchange factor matrix, and let $\mathbf{A}$ and $\mathbf{R}$ be the absorption and reflection-scattering matrices defined by the exchange factor formulation. Then:
\begin{enumerate}
\item The sum $\mathbf{A} + \mathbf{R}$ can be expressed as
\begin{equation}
\mathbf{A} + \mathbf{R} = \mathbf{F}
\label{eq:A_plus_R_form}
\end{equation}
\item The spectral radius satisfies $\rho(\mathbf{A} + \mathbf{R}) = 1$
\item The vector $\mathbf{v} = \mathbf{1}$ is the Perron eigenvector corresponding to the positive real eigenvalue $\lambda = 1$
\end{enumerate}
\end{theorem}

\begin{proof}
The result is established through two key relationships.

\textbf{Step 1:} Recognizing $\mathbf{A}+\mathbf{R}=\mathbf{F}$:
\begin{align}
	\mathbf{A} + \mathbf{R} &= \mathbf{F} \circ (\mathbf{1}-\mathbf{B})+\mathbf{F} \circ \mathbf{B} = \mathbf{F}
\label{eq:A_plus_R_is_F}
\end{align}

\textbf{Step 2:} Combining with $\mathbf{F}\mathbf{1}=\mathbf{1}$ yields the final eigenvalue identity:
\begin{equation}
(\mathbf{A} + \mathbf{R})\mathbf{1} = \mathbf{F}\mathbf{1} = \mathbf{1}
\end{equation}

Therefore, $\mathbf{A}+\mathbf{R}$ is row stochastic, and $\mathbf{v} = \mathbf{1}$ is an eigenvector with real eigenvalue $\lambda = 1$. Since all components of $\mathbf{v}$ are positive, this is the Perron eigenvector. By the Perron-Frobenius theorem, the corresponding eigenvalue equals the spectral radius \citep{HornJohnson2018}, hence $\rho(\mathbf{A} + \mathbf{R}) = 1$.
\end{proof}

\begin{theorem}[Upper Bound on Spectral Radius of the Reflection-Scattering Matrix]
\label{thm:reflection_upper_bound}
Let $\mathbf{F}$ be non-negative and row stochastic, and let $0 \leq b_k \leq 1$ for all $k$. Then the spectral radius of the reflection-scattering matrix $\mathbf{R} = \mathbf{F} \circ \mathbf{B}$ satisfies
\begin{equation}
\rho(\mathbf{R}) \leq b_{\max},
\label{eq:rho_R_upper}
\end{equation}
where $b_{\max} = \max_j b_j$, with equality when $b_j = b_{\max}$ uniformly. Consequently, $\mathbf{D} = \mathbf{I} - \mathbf{R}^\top$ is invertible whenever $b_{\max} < 1$.
\end{theorem}

\begin{proof}
Since $\mathbf{B}$ is column-constant with column $j$ equal to $b_j$, the entries of $\mathbf{R}$ are $R_{ij} = b_j F_{ij} \geq 0$, and each row sum satisfies
\begin{equation*}
\sum_j R_{ij} = \sum_j b_j F_{ij} \leq b_{\max} \sum_j F_{ij} = b_{\max},
\end{equation*}
using row-stochasticity of $\mathbf{F}$ in the last step. Therefore $\|\mathbf{R}\|_\infty \leq b_{\max}$, and the bound \eqref{eq:rho_R_upper} follows from $\rho(\mathbf{R}) \leq \|\mathbf{R}\|_\infty$ \citep{HornJohnson2018}. When $b_j = b_{\max}$ uniformly, $\mathbf{R} = b_{\max} \mathbf{F}$ gives $\rho(\mathbf{R}) = b_{\max}\, \rho(\mathbf{F}) = b_{\max}$.
\end{proof}

\begin{corollary}[Conditioning of $\mathbf{D}$]
\label{cor:D_conditioning}
Let $\mathbf{F}$ be non-negative and row stochastic, and let $0 \leq b_k \leq b_{\max} < 1$ for all $k$. Then the condition number of $\mathbf{D} = \mathbf{I} - \mathbf{R}^\top$ in the 1-norm satisfies
\begin{equation}
\operatorname{cond}_1(\mathbf{D}) = \|\mathbf{D}\|_1 \|\mathbf{D}^{-1}\|_1 \leq \frac{1 + b_{\max}}{1 - b_{\max}} .
\label{eq:D_conditioning}
\end{equation}
If $b_k = b$ uniformly and $F_{kk} = 0$ for at least one element $k$, equality holds.
\end{corollary}

\begin{proof}
The 1-norm of $\mathbf{R}^\top$ is the largest column sum of $\mathbf{R}^\top$, which is the largest row sum of $\mathbf{R}$, so $\|\mathbf{R}^\top\|_1 = \|\mathbf{R}\|_\infty \leq b_{\max}$ by the proof of Theorem~\ref{thm:reflection_upper_bound}. Hence $\|\mathbf{D}\|_1 \leq 1 + b_{\max}$, and since $\|\mathbf{R}^\top\|_1 < 1$, the Neumann series $\mathbf{D}^{-1} = \sum_{p=0}^{\infty} (\mathbf{R}^\top)^p$ gives $\|\mathbf{D}^{-1}\|_1 \leq 1/(1 - b_{\max})$. Multiplying the two bounds gives equation \eqref{eq:D_conditioning}.

For uniform $b$, $\mathbf{R}^\top = b\,\mathbf{F}^\top$ and $\mathbf{1}^\top \mathbf{D} = \mathbf{1}^\top - b\,(\mathbf{F}\mathbf{1})^\top = (1-b)\,\mathbf{1}^\top$, so every column of $\mathbf{D}^{-1}$ sums to $1/(1-b)$.

Since $D_{ik} = \delta_{ik} - b F_{ki}$, with $\delta_{ik}$ the Kronecker delta, column $k$ of $\mathbf{D}$ contains row $k$ of $\mathbf{F}$: its diagonal entry is $1 - bF_{kk} > 0$, and its off-diagonal entries have absolute values $bF_{ki}$, which sum to $b\,(1 - F_{kk})$ by row-stochasticity. The absolute values in column $k$ therefore sum to $(1 - b F_{kk}) + b\,(1 - F_{kk}) = 1 + b - 2bF_{kk}$, which is at most $1 + b$ and equals $1 + b$ when $F_{kk} = 0$, so $\|\mathbf{D}\|_1 = 1 + b$.
\end{proof}

The bound \eqref{eq:D_conditioning} is independent of the number of elements, the geometry and the optical thickness: when all emissive powers are prescribed, the linear system is well conditioned unless some element approaches perfect reflection or pure scattering. Equality holds, for example, in transparent enclosures with flat surfaces and uniform reflectivity, for which $F_{kk} = 0$. 

\subsection{A Discrete Minimum Principle}
\label{sec:min_principle}

Two questions decide whether the mixed-boundary system \eqref{eq:solve_for_j} is physically meaningful: whether it has a unique solution, and whether non-negative prescribed quantities $\mathbf{h} \geq \mathbf{0}$ always give a non-negative solution $\mathbf{j} \geq \mathbf{0}$. Both are answered by a single argument. To formulate it, let $S$ denote the set of elements whose emissive power is prescribed, i.e. the rows of $\mathbf{M}$ taken from $\mathbf{D}$, and define the masked reflection-scattering coefficients
\begin{equation}
\tilde{b}_k =
\begin{cases}
b_k, & k \in S, \\
1, & k \notin S.
\end{cases}
\label{eq:masked_albedo}
\end{equation}
Since $D_{ki} = \delta_{ki} - b_k F_{ik}$ and $C_{ki} = \delta_{ki} - F_{ik}$, every row of $\mathbf{M}$ has the same form, and
\begin{equation}
\mathbf{M} = \mathbf{I} - \operatorname{diag}(\tilde{\mathbf{b}})\,\mathbf{F}^\top .
\label{eq:M_masked}
\end{equation}
Physically, elements with prescribed emissive power retain their true reflection-scattering coefficient, while elements with prescribed non-negative source term re-emit all incident power, as in radiative equilibrium.

Both questions concern the transpose $\mathbf{M}^\top = \mathbf{I} - \mathbf{F}\operatorname{diag}(\tilde{\mathbf{b}})$, which brings $\mathbf{F}$ into its original orientation, whose rows sum to unity. The first asks whether $\mathbf{M}^\top\mathbf{x} = \mathbf{0}$ has only the trivial solution, since $\mathbf{M}$ is non-singular if and only if $\mathbf{M}^\top$ is. The second asks whether $\mathbf{M}^{-1} \geq \mathbf{0}$ entrywise, which is equivalent to $\mathbf{M}^{-\top} \geq \mathbf{0}$, i.e. to every solution of $\mathbf{M}^\top\mathbf{x} = \mathbf{c}$ with $\mathbf{c} \geq \mathbf{0}$ being non-negative. Written entry by entry, $\mathbf{M}^\top\mathbf{x} = \mathbf{c}$ reads
\begin{equation}
x_k = \sum_l F_{kl}\, \tilde{b}_l\, x_l + c_k \quad \text{for all } k .
\label{eq:masked_fixed}
\end{equation}
Since the rows of $\mathbf{F}$ are non-negative and sum to unity, $x_k$ is a weighted average of the masked values $\tilde{b}_l x_l$ of the elements with which $k$ exchanges directly, plus the non-negative term $c_k$. 
Equation \eqref{eq:masked_fixed} is the Markov chain of Naraghi and Chung \citep{Naraghi1985}, whose transition probabilities are the exchange factors $\mathbf{F}$ and whose equivalent reflection matrix is $\operatorname{diag}(\mathbf{b})$, with $\tilde{\mathbf{b}}$ in place of $\mathbf{b}$: a ray emitted from element $k$ has its first interaction with element $l$ with probability $F_{kl}$, is re-emitted there with probability $\tilde{b}_l$ and is absorbed otherwise. Reading equation \eqref{eq:masked_fixed} as $x_k = c_k + \sum_l F_{kl}\, \tilde{b}_l\, x_l$, the quantity $x_k$ is the expected sum of the entries $c_m$ of $\mathbf{c}$ over all elements $m$ from which such a ray is emitted or re-emitted before it is absorbed; when $\mathbf{c}$ is the $i$-th unit vector, it is the expected number of emissions from element $i$, which is the entry $(\mathbf{M}^{-1})_{ik}$. Such an expectation cannot be negative. The constant vector shows what is required for this: if $\tilde{b}_k = 1$ for every $k$, then $\mathbf{x} = -\mathbf{1}$ solves equation \eqref{eq:masked_fixed} with $\mathbf{c} = \mathbf{0}$, since $\mathbf{F}\mathbf{1} = \mathbf{1}$. A single element that absorbs, $\tilde{b}_{k_0} < 1$, rules this out, as the following lemma shows without probabilistic arguments. Its proof is by contradiction, in three steps: a negative minimum of $\mathbf{x}$ forces equality in the weighted average of equation \eqref{eq:masked_fixed} at every element where the minimum is attained; equality forces every element that such an element exchanges with to attain the minimum too, with $\tilde{b}_l = 1$; and irreducibility then carries the minimum to an element that exchanges directly with $k_0$, which is incompatible with $\tilde{b}_{k_0} < 1$.

\begin{lemma}[Masked Minimum Principle]
\label{lem:min_principle}
Let $\mathbf{F}$ be non-negative, row stochastic and irreducible, let $0 \leq \tilde{b}_k \leq 1$ for all $k$, and suppose that $\tilde{b}_{k_0} < 1$ for at least one element $k_0$. If a vector $\mathbf{x}$ satisfies equation \eqref{eq:masked_fixed} with $\mathbf{c} \geq \mathbf{0}$, then $\mathbf{x} \geq \mathbf{0}$.
\end{lemma}

\begin{proof}
Suppose $\mathbf{x}$ has a negative entry, let $\mu = \min_k x_k < 0$, and let $P$ be the set of elements with $x_k = \mu$. For every element $l$, $\tilde{b}_l x_l \geq \mu$: if $x_l \geq 0$, then $\tilde{b}_l x_l \geq 0 > \mu$, and if $x_l < 0$, then $\tilde{b}_l x_l \geq x_l \geq \mu$, since $\tilde{b}_l \leq 1$. For any $k \in P$, equation \eqref{eq:masked_fixed}, row-stochasticity and $c_k \geq 0$ therefore give
\begin{equation*}
\mu = x_k = \sum_l F_{kl}\, \tilde{b}_l\, x_l + c_k \geq \sum_l F_{kl}\, \mu + c_k = \mu + c_k \geq \mu .
\end{equation*}
Equality therefore holds throughout, so $c_k = 0$ and $\tilde{b}_l x_l = \mu$ for every $l$ with $F_{kl} > 0$. Since $\mu < 0$, this requires $x_l < 0$, and then $\mu \leq x_l \leq \tilde{b}_l x_l = \mu$ gives $x_l = \mu$ and $\tilde{b}_l = 1$. Hence $P$ is closed under direct exchange: if $k \in P$ and $F_{kl} > 0$, then $l \in P$ and $\tilde{b}_l = 1$.

By irreducibility, element $k_0$ receives direct exchange from at least one element $i$, i.e. $F_{ik_0} > 0$, and $i$ is reachable from any element of $P$ through a chain of direct exchanges. Following this chain from an element of $P$ shows that $i \in P$, and then $F_{ik_0} > 0$ gives $\tilde{b}_{k_0} = 1$, contradicting $\tilde{b}_{k_0} < 1$.
\end{proof}

\begin{theorem}[Non-singularity of Mixed Boundaries]
\label{thm:M_nonsingular}
Let $\mathbf{F}$ be non-negative, row stochastic and irreducible, and let $\mathbf{M}$ be assembled with the rows in $S$ taken from $\mathbf{D}$ and the remaining rows taken from $\mathbf{C}$. If $0 \leq b_k \leq 1$ for all $k \in S$, and $b_{k_0} < 1$ for at least one $k_0 \in S$, then $\mathbf{M}$ is non-singular.
\end{theorem}

\begin{proof}
It suffices to show that $\mathbf{M}^\top \mathbf{x} = \mathbf{0}$ implies $\mathbf{x} = \mathbf{0}$. By equation \eqref{eq:M_masked}, $\mathbf{M}^\top \mathbf{x} = \mathbf{0}$ is equation \eqref{eq:masked_fixed} with $\mathbf{c} = \mathbf{0}$. Moreover, $0 \leq \tilde{b}_k \leq 1$ for all $k$, and $\tilde{b}_{k_0} = b_{k_0} < 1$ since $k_0 \in S$. By Lemma~\ref{lem:min_principle}, $\mathbf{x} \geq \mathbf{0}$, and since $-\mathbf{x}$ satisfies the same equation, also $-\mathbf{x} \geq \mathbf{0}$. Hence $\mathbf{x} = \mathbf{0}$, and $\mathbf{M}$ is non-singular.
\end{proof}

The hypothesis $b_{k_0} < 1$ for some $k_0 \in S$ cannot be dropped: if $b_k = 1$ for all $k \in S$, every row taken from $\mathbf{D}$ equals the row of $\mathbf{C}$ it replaces, so $\mathbf{M} = \mathbf{C}$, which is singular since $\mathbf{C}^\top\mathbf{1} = \mathbf{1} - \mathbf{F}\mathbf{1} = \mathbf{0}$. When all emissive powers are prescribed, $S$ contains every element and $\mathbf{M} = \mathbf{D}$, so Theorem~\ref{thm:M_nonsingular} also sharpens Theorem~\ref{thm:reflection_upper_bound}: for irreducible $\mathbf{F}$, a single element with $b_{k_0} < 1$ suffices for $\mathbf{D}$ to be non-singular, which covers, for example, a purely scattering medium ($\omega = 1$) enclosed by absorbing walls. For this case, Noble \citep{Noble1975} gives sufficient conditions for the existence of the inverse matrices of his explicit relations, by irreducible diagonal dominance; Theorem~\ref{thm:M_nonsingular} confirms the invertibility by a different route, and extends it to prescribed source terms. Lemma~\ref{lem:min_principle} and Theorem~\ref{thm:M_nonsingular} have been machine-checked in Lean 4 \citep{Bielefeld2024}.

The irreducibility of $\mathbf{F}$ is crucial for these results, and cannot be dropped either. Physically, $\mathbf{F}$ is irreducible if and only if the domain has no isolated regions, that is, for any two elements $i$ and $j$ there exists a chain of direct exchanges connecting them. This connectivity condition is the mathematical irreducibility required by Lemma~\ref{lem:min_principle}.

\section{Physical Correctness}
\label{sec:physical_correct}

For a physical model to be useful it should possess physical correctness. The conditions which guarantee physical correctness of results calculated with the proposed exchange factor formulation are described in the following subsections.

\subsection{Non-Negative Radiation}

For radiative transfer, physical correctness means that the calculated rate of radiation must be non-negative, since negative radiation is not a physically meaningful quantity.

\begin{theorem}[Non-negative Radiation]
\label{thm:physical_correctness}
Let $\mathbf{F}$ be non-negative, row stochastic and irreducible, let $0 \leq b_k \leq 1$ for all $k$, and let $\mathbf{M}$ be assembled as in Theorem~\ref{thm:M_nonsingular}, with $b_{k_0} < 1$ for at least one $k_0 \in S$. Then $\mathbf{M}^{-1} \geq \mathbf{0}$ entrywise, and for any non-negative right-hand side $\mathbf{h} \geq \mathbf{0}$, the unique solution of $\mathbf{M}\mathbf{j} = \mathbf{h}$ satisfies:
\begin{enumerate}
\item Non-negative total radiant power: $\mathbf{j} = \mathbf{M}^{-1}\mathbf{h} \geq \mathbf{0}$
\item Non-negative reflected-scattered power: $\mathbf{r} = \mathbf{R}^\top\mathbf{j} \geq \mathbf{0}$
\item Non-negative emissive power: $\mathbf{e} = \mathbf{D}\mathbf{j} \geq \mathbf{0}$
\end{enumerate}
\end{theorem}

\begin{proof}
Part 1: Non-negative total radiant power.
By Theorem~\ref{thm:M_nonsingular}, $\mathbf{M}$ is non-singular. For each element $i$, let $\mathbf{x}$ be the solution of $\mathbf{M}^\top\mathbf{x} = \mathbf{e}_i$, where $\mathbf{e}_i$ is the $i$-th unit vector, so that $\mathbf{x}$ is the $i$-th column of $\mathbf{M}^{-\top}$, i.e. the $i$-th row of $\mathbf{M}^{-1}$. By equation \eqref{eq:M_masked}, $\mathbf{M}^\top\mathbf{x} = \mathbf{e}_i$ is equation \eqref{eq:masked_fixed} with $\mathbf{c} = \mathbf{e}_i \geq \mathbf{0}$, and since $0 \leq \tilde{b}_k \leq 1$ for all $k$ and $\tilde{b}_{k_0} = b_{k_0} < 1$, Lemma~\ref{lem:min_principle} gives $\mathbf{x} \geq \mathbf{0}$. Every row of $\mathbf{M}^{-1}$ is therefore non-negative, so $\mathbf{M}^{-1} \geq \mathbf{0}$, and for $\mathbf{h} \geq \mathbf{0}$:
\begin{equation}
\mathbf{j} = \mathbf{M}^{-1}\mathbf{h} \geq \mathbf{0}
\end{equation}

Part 2: Non-negative reflected-scattered power.
Since $\mathbf{F} \geq \mathbf{0}$ and $b_k \geq 0$, $\mathbf{R} \geq \mathbf{0}$, and with $\mathbf{j} \geq \mathbf{0}$ from Part 1:
\begin{equation}
\mathbf{r} = \mathbf{R}^\top\mathbf{j} \geq \mathbf{0}
\end{equation}

Part 3: Non-negative emissive power.
Since $\mathbf{D} = \mathbf{C} + \mathbf{A}^\top$, the emissive power and the source term are related by
\begin{equation}
\mathbf{e} = \mathbf{D}\mathbf{j} = \mathbf{C}\mathbf{j} + \mathbf{A}^\top\mathbf{j} = \mathbf{q} + \mathbf{A}^\top\mathbf{j} .
\label{eq:emissive_equivalence}
\end{equation}
For $k \in S$, the emissive power is prescribed, so $e_k = h_k \geq 0$. For $k \notin S$, the source term is prescribed, so $q_k = h_k \geq 0$, and since $\mathbf{A} \geq \mathbf{0}$ (as $b_k \leq 1$) and $\mathbf{j} \geq \mathbf{0}$, $e_k = q_k + (\mathbf{A}^\top\mathbf{j})_k \geq 0$.
\end{proof}

\subsection{Conservation of Energy}

At the steady state, energy conservation requires that the source fluxes $\mathbf{q}=\mathbf{C}\mathbf{j}$ calculated from any total radiant power solution vector $\mathbf{j}$ sum to a scalar of zero: $\mathbf{1}^\top\mathbf{C}\mathbf{j} = 0$. This fundamental physical principle holds unconditionally for the proposed exchange factor formulation.

\begin{theorem}[Unconditional Energy Conservation]\label{thm:unconditional_conservation}
Let $\mathbf{F}$ be the row-stochastic exchange factor matrix, and let $\mathbf{b}$ be the vector of coefficients for the column constant single reflection-scattering matrix $\mathbf{B}$, and define the remaining system matrices as given by the proposed exchange factor formulation.

Then for any mixed system $\mathbf{M}\mathbf{j} = \mathbf{h}$ where $\mathbf{M}$ is assembled from rows of $\mathbf{C}$ and $\mathbf{D}$, the sum of the source fluxes equals zero, meaning the energy conservation property holds:
\begin{equation}
\mathbf{1}^\top \mathbf{C} \mathbf{j} = 0
\end{equation}
regardless of the choice of rows or the signs of the boundary conditions in the mixed system.
\end{theorem}

\begin{proof}
Since $\mathbf{A} + \mathbf{R} = \mathbf{F}$ by Theorem~\ref{thm:spectral_radius_unity}, $\mathbf{C} = \mathbf{I} - \mathbf{F}^\top$, and row-stochasticity, $\mathbf{F}\mathbf{1} = \mathbf{1}$, gives $\mathbf{1}^\top \mathbf{C} = \mathbf{1}^\top - (\mathbf{F}\mathbf{1})^\top = \mathbf{0}^\top$. Therefore, for any vector $\mathbf{j}$:
\begin{equation}
\mathbf{1}^\top \mathbf{C}\mathbf{j} = \mathbf{0}^\top\mathbf{j} = 0
\end{equation}
In particular, this holds for the solution of any mixed system, regardless of how it is constructed.
\end{proof}

\subsection{Albedo Invariance}
\label{sec:albedo_invariance}

A grey medium in radiative equilibrium absorbs and emits equal amounts of energy at every point, so its source function equals the mean incident intensity irrespective of the single scattering albedo, and the solution to the radiative transfer equation depends on the extinction coefficient alone \citep{Modest2022}. In zonal analysis, Larsen \citep{Larsen1983} recognised that isotropic scattering and absorption followed by re-emission are indistinguishable in their contribution to the paths of energy between zones, so that his exchange factors depend on the total extinction only. The proposed formulation inherits this property as an algebraic identity.

\begin{proposition}[Albedo Invariance]
\label{prop:albedo_invariance}
Let $\mathbf{F}$ be fixed, and let $\mathbf{M}$ be assembled as in Theorem~\ref{thm:M_nonsingular}, with the source term prescribed for every element $k \notin S$. Then $\mathbf{M}$ and $\mathbf{h}$ do not depend on $b_k$ for any $k \notin S$, and neither do the solution $\mathbf{j}$, the incident power $\mathbf{g} = \mathbf{F}^\top\mathbf{j}$ and the source terms $\mathbf{q} = \mathbf{C}\mathbf{j}$. If moreover $q_k = 0$ and $b_k < 1$ for an element $k \notin S$, its blackbody emissive power
\begin{equation}
E_{\mathrm{b},k} = \frac{e_k}{(1 - b_k)\, w_k} = \frac{g_k}{w_k}
\label{eq:albedo_invariance}
\end{equation}
does not depend on $b_k$ either.
\end{proposition}

\begin{proof}
By equation \eqref{eq:M_masked}, $\mathbf{M} = \mathbf{I} - \operatorname{diag}(\tilde{\mathbf{b}})\,\mathbf{F}^\top$ with $\tilde{b}_k = 1$ for every $k \notin S$, so the coefficients $b_k$ with $k \notin S$ do not enter $\mathbf{M}$. The right-hand side $\mathbf{h}$ contains the prescribed source terms $q_k$ for $k \notin S$ and the prescribed emitted powers for $k \in S$, neither of which involves $b_k$ for $k \notin S$. Since $\mathbf{M}$ is non-singular by Theorem~\ref{thm:M_nonsingular}, the solution $\mathbf{j} = \mathbf{M}^{-1}\mathbf{h}$ is the same for all values of these coefficients, and so are $\mathbf{g} = \mathbf{F}^\top\mathbf{j}$ and $\mathbf{q} = \mathbf{C}\mathbf{j}$, since $\mathbf{C} = \mathbf{I} - \mathbf{F}^\top$ involves no reflection-scattering coefficients. For an element $k \notin S$ with $q_k = 0$, equation \eqref{eq:emissive_equivalence} gives $e_k = (\mathbf{A}^\top\mathbf{j})_k = (1 - b_k) \sum_i F_{ik}\, j_i = (1 - b_k)\, g_k$, and dividing by $(1 - b_k)\, w_k$ as in Remark~\ref{rem:temperature_recovery} gives equation \eqref{eq:albedo_invariance}.
\end{proof}

A medium in radiative equilibrium yields identical radiant powers and temperatures whether it absorbs and emits, scatters, or does both, and a purely scattering medium ($b_k = 1$) yields the same radiant powers as an absorbing-emitting one. In the proposed formulation this follows by inspection of $\mathbf{M}$, before any system is solved; in Noble's formulation \citep{Noble1975} the single scattering albedo enters the intermediate inverses $\mathbf{P}$ and $\mathbf{R}$ and cancels only in the final result, as Larsen and Howell showed when transforming Noble's total exchange areas to their exchange factors \citep{LarsenHowell1985}. The symbolic example of section \ref{sec:radiant_power_albedo} illustrates the proposition.

\section{Equivalence with Hottel's Zonal Method}
\label{sec:hottel_equivalence}

Hottel's zonal method \citep{Hottel1958, Hottel1967} describes radiative exchange between surface and volume zones through energy balances on area-specific quantities, and Noble \citep{Noble1975} recast it in matrix form by eliminating the radiosities to obtain total exchange areas. This section shows that the proposed formulation and Hottel's zonal balances are the same discrete model, expressed in different variables: the proposed formulation works with total radiant powers $\mathbf{j}$ in Watts, while Hottel's method works with area-specific radiosities. The equivalence requires no property of $\mathbf{F}$, not even reciprocity, and holds for arbitrary reflectivities and single scattering albedos of the individual elements. As a consequence, the guarantees established in Sections~\ref{sec:matrix_props} and \ref{sec:physical_correct} apply to Hottel's balances as well, while the proposed formulation obtains the solution from a single linear solve whose sparsity follows that of $\mathbf{F}$, without forming total exchange areas. Reciprocity enters only when the direct exchange areas are written as a symmetric matrix, as in the closed-form total exchange areas of Noble.

\subsection{Hottel's Zonal Balances}
\label{sec:hottel_balances}

The comparison is made under the assumptions shared by both formulations: grey radiation, diffuse emission and reflection at surfaces, isotropic scattering in the medium, and uniform properties and radiosity within each zone. In Hottel's method, each zone $k$ is characterized by its radiosity $J_k$, the power leaving the zone per unit of emission capacity $w_k$, and by its blackbody emissive power $E_{\mathrm{b},k} = n_k^2 \sigma T_k^4$. The direct exchange areas are the entries of
\begin{equation}
\mathbf{W}\mathbf{F} =
\begin{bmatrix}
\overline{\mathbf{ss}} & \overline{\mathbf{sg}} \\
\overline{\mathbf{gs}} & \overline{\mathbf{gg}}
\end{bmatrix},
\end{equation}
which is symmetric when $\mathbf{F}$ is reciprocal, equation \eqref{eq:reciprocity}.

The power incident on zone $k$ is $g_k = \sum_i w_i F_{ik} J_i$, and the power leaving zone $k$ is the sum of its emitted power and its reflected or scattered incident power:
\begin{equation}
w_k J_k = e_k + b_k g_k ,
\label{eq:hottel_balance}
\end{equation}
where the emitted power is
\begin{equation}
e_k = (1 - b_k)\, w_k\, E_{\mathrm{b},k} ,
\label{eq:hottel_emission}
\end{equation}
which equals $\varepsilon_k A_k E_{\mathrm{b},k}$ for a surface zone and $4 \kappa_k V_k E_{\mathrm{b},k}$ for a volume zone, since $\kappa_k = (1-\omega_k)\beta_k$. The net source term of zone $k$ is its emitted power minus its absorbed power:
\begin{equation}
q_k = e_k - (1 - b_k) g_k = w_k J_k - g_k .
\label{eq:hottel_source}
\end{equation}
Hottel and Sarofim \citep{Hottel1967} define the direct exchange area through the one-way flux from the emitter, $\overline{s_i s_k} = A_i F_{ik}$ for the power leaving $i$ and arriving at $k$, so the incident power on zone $k$ is $\sum_i \overline{s_i s_k} J_i = \sum_i w_i F_{ik} J_i = g_k$; reciprocity, $A_i F_{ik} = A_k F_{ki}$, then makes the areas symmetric, so that the entry for the pair $(i,k)$ may equally be taken from the row of the receiver $k$.

\begin{theorem}[Equivalence with Hottel's Zonal Balances]
\label{thm:hottel_equivalence}
Let $\mathbf{W} = \operatorname{diag}(\mathbf{w})$ with $w_k > 0$ for all $k$, and let the direct exchange areas be defined from the emitter, $\overline{s_i s_k} = w_i F_{ik}$, as by Hottel and Sarofim \citep{Hottel1967}. For any vectors $\mathbf{j}$, $\mathbf{e}$ and $\mathbf{q}$, let $\mathbf{J} = \mathbf{W}^{-1}\mathbf{j}$. Then
\begin{equation}
\mathbf{e} = \mathbf{D}\mathbf{j}
\quad \text{and} \quad
\mathbf{q} = \mathbf{C}\mathbf{j}
\end{equation}
if and only if $\mathbf{J}$, $\mathbf{e}$ and $\mathbf{q}$ satisfy Hottel's zonal balances, equations \eqref{eq:hottel_balance} and \eqref{eq:hottel_source}, for every zone $k$. Moreover, the incident power of Hottel's balances is $\mathbf{g} = \mathbf{F}^\top\mathbf{j} = \mathbf{A}^\top\mathbf{j} + \mathbf{R}^\top\mathbf{j}$, as in equation \eqref{eq:r_and_g}. No property of $\mathbf{F}$ is required; in particular, the equivalence holds without reciprocity.
\end{theorem}

\begin{proof}
Since $\mathbf{W}$ is diagonal with positive entries, $\mathbf{j} = \mathbf{W}\mathbf{J}$ and $\mathbf{J} = \mathbf{W}^{-1}\mathbf{j}$ correspond one-to-one, with $j_k = w_k J_k$. The incident power of zone $k$ is therefore
\begin{equation*}
g_k = \sum_i \overline{s_i s_k}\, J_i = \sum_i w_i F_{ik} J_i = \sum_i F_{ik}\, j_i = (\mathbf{F}^\top \mathbf{j})_k ,
\end{equation*}
and since $\mathbf{A} + \mathbf{R} = \mathbf{F}$ by Theorem~\ref{thm:spectral_radius_unity}, $\mathbf{g} = \mathbf{A}^\top\mathbf{j} + \mathbf{R}^\top\mathbf{j}$. Because $\mathbf{B}$ is column-constant, the reflected or scattered power of zone $k$ is
\begin{equation*}
(\mathbf{R}^\top \mathbf{j})_k = \sum_i F_{ik}\, b_k\, j_i = b_k\, g_k .
\end{equation*}
Hence $(\mathbf{D}\mathbf{j})_k = j_k - b_k g_k$ and $(\mathbf{C}\mathbf{j})_k = j_k - g_k$. Equation \eqref{eq:hottel_balance}, $w_k J_k = e_k + b_k g_k$, is therefore equivalent to $e_k = (\mathbf{D}\mathbf{j})_k$, and equation \eqref{eq:hottel_source}, $q_k = w_k J_k - g_k$, is equivalent to $q_k = (\mathbf{C}\mathbf{j})_k$. Since this holds for every zone $k$, the two systems are equivalent.
\end{proof}

\begin{remark}[Non-reciprocal exchange factors]
\label{rem:nonreciprocal}
None of the results of Sections~\ref{sec:matrix_props} and \ref{sec:physical_correct} assumes reciprocity: non-singularity, non-negativity and unconditional energy conservation hold for any non-negative, row-stochastic and irreducible $\mathbf{F}$, including the unsmoothed ray-traced exchange factors used in Section~\ref{sec:numerical_validation}, and by Theorem~\ref{thm:hottel_equivalence} they hold for Hottel's balances with the same exchange factors. Reciprocity is required only where the direct exchange areas are used as a symmetric matrix, $\overline{s_i s_k} = \overline{s_k s_i}$, as in the closed-form total exchange areas of Section~\ref{sec:total_exchange_areas}. The script of Section~\ref{sec:hottel_numerical} verifies the agreement with unsmoothed exchange factors.
\end{remark}

\begin{corollary}[Properties of Hottel's Zonal Balances]
\label{cor:hottel_properties}
Let $\mathbf{F}$ be non-negative, row stochastic and irreducible, let $0 \leq b_k \leq 1$ for all $k$, and consider Hottel's zonal balances, equations \eqref{eq:hottel_balance} and \eqref{eq:hottel_source}, with the direct exchange areas $\overline{s_i s_k} = w_i F_{ik}$, the emitted power $e_k$ prescribed for $k \in S$ and the source term $q_k$ prescribed for $k \notin S$. If $b_{k_0} < 1$ for at least one $k_0 \in S$, then:
\begin{enumerate}
\item The balances have a unique solution $\mathbf{J}$
\item If all prescribed quantities are non-negative, then $\mathbf{J} \geq \mathbf{0}$, $\mathbf{g} \geq \mathbf{0}$ and $\mathbf{e} \geq \mathbf{0}$
\item The source terms satisfy $\mathbf{1}^\top \mathbf{q} = 0$
\end{enumerate}
\end{corollary}

\begin{proof}
By Theorem~\ref{thm:hottel_equivalence}, $\mathbf{J}$ solves the balances with the prescribed quantities if and only if $\mathbf{j} = \mathbf{W}\mathbf{J}$ solves $\mathbf{M}\mathbf{j} = \mathbf{h}$, where $\mathbf{M}$ takes the rows in $S$ from $\mathbf{D}$ and the remaining rows from $\mathbf{C}$, and $\mathbf{h}$ contains the prescribed quantities. Part 1 then follows from Theorem~\ref{thm:M_nonsingular}. For Part 2, Theorem~\ref{thm:physical_correctness} gives $\mathbf{j} \geq \mathbf{0}$ and $\mathbf{e} \geq \mathbf{0}$, so $\mathbf{J} = \mathbf{W}^{-1}\mathbf{j} \geq \mathbf{0}$ and $\mathbf{g} = \mathbf{F}^\top\mathbf{j} \geq \mathbf{0}$. Part 3 follows from Theorem~\ref{thm:unconditional_conservation}, since $\mathbf{q} = \mathbf{C}\mathbf{j}$.
\end{proof}

\begin{remark}[Temperature recovery]
\label{rem:temperature_recovery}
Theorem~\ref{thm:hottel_equivalence} is stated in powers and requires no conversion to temperatures. Temperatures follow from equation \eqref{eq:hottel_emission} for every zone with $b_k < 1$:
\begin{equation*}
E_{\mathrm{b},k} = \frac{e_k}{(1-b_k)\, w_k} ,
\qquad
T_k = \left( \frac{E_{\mathrm{b},k}}{n_k^2 \sigma} \right)^{1/4} ,
\end{equation*}
which reproduces equation \eqref{eq:temperature}. For a zone with $b_k = 1$, i.e. a perfectly reflecting surface or a purely scattering volume, $e_k = 0$, and the temperature is not determined by radiation; it must be obtained from other heat transfer modes.
\end{remark}

\subsection{Total Exchange Areas}
\label{sec:total_exchange_areas}

Noble's matrix formulation \citep{Noble1975} eliminates the radiosities from Hottel's balances for the case where all temperatures are prescribed, and expresses the absorbed power of each zone directly in terms of the blackbody emissive powers $\mathbf{E}_\mathrm{b}$ through the total exchange areas $\overline{\mathbf{SS}}$, $\overline{\mathbf{SG}}$, $\overline{\mathbf{GS}}$ and $\overline{\mathbf{GG}}$. Theorem~\ref{thm:hottel_equivalence} recovers these quantities from the proposed formulation.

\begin{proposition}[Total Exchange Areas]
\label{prop:total_exchange_areas}
Let $\mathbf{F}$ be non-negative, row stochastic and irreducible, let $0 \leq b_k \leq 1$ for all $k$ with $b_{k_0} < 1$ for at least one $k_0$, and let all emitted powers be prescribed, $\mathbf{e} = (\mathbf{I} - \operatorname{diag}(\mathbf{b}))\,\mathbf{W}\,\mathbf{E}_\mathrm{b}$. Then the absorbed powers are $\mathbf{g}_\mathrm{a} = \mathbf{A}^\top\mathbf{j}$ with $\mathbf{j} = \mathbf{D}^{-1}\mathbf{e}$, and the matrix that maps the blackbody emissive powers to the absorbed powers is
\begin{equation}
\mathbf{T} = \mathbf{A}^\top \mathbf{D}^{-1} \left(\mathbf{I} - \operatorname{diag}(\mathbf{b})\right) \mathbf{W} .
\label{eq:total_exchange_areas}
\end{equation}
\end{proposition}

\begin{proof}
With all emitted powers prescribed, $\mathbf{M} = \mathbf{D}$, which is non-singular by Theorem~\ref{thm:M_nonsingular}, so $\mathbf{j} = \mathbf{D}^{-1}\mathbf{e}$ is the unique solution, and by Theorem~\ref{thm:hottel_equivalence} it corresponds to the unique solution of Hottel's balances. The absorbed power is $\mathbf{g}_\mathrm{a} = \mathbf{A}^\top\mathbf{j} = \mathbf{A}^\top\mathbf{D}^{-1}(\mathbf{I} - \operatorname{diag}(\mathbf{b}))\,\mathbf{W}\,\mathbf{E}_\mathrm{b}$.
\end{proof}

\begin{proposition}[Positivity of the Total Exchange Areas]
\label{prop:total_exchange_positivity}
Under the hypotheses of Proposition~\ref{prop:total_exchange_areas}, if $0 < b_k < 1$ for all $k$, then every entry of $\mathbf{T}$ is strictly positive, regardless of the sparsity of $\mathbf{F}$.
\end{proposition}

\begin{proof}
Let $\mathbf{e}_i$ be the $i$-th unit vector, and let $\mathbf{x}$ be the $i$-th row of $\mathbf{D}^{-1}$, i.e. the solution of $\mathbf{D}^\top\mathbf{x} = \mathbf{e}_i$. With all emitted powers prescribed, $\tilde{b}_k = b_k$ for every $k$, so $\mathbf{x}$ satisfies equation \eqref{eq:masked_fixed} with $\mathbf{c} = \mathbf{e}_i$, and $\mathbf{x} \geq \mathbf{0}$ by Lemma~\ref{lem:min_principle}, as in Theorem~\ref{thm:physical_correctness}. It remains to exclude zero entries, which is where $b_l > 0$ enters. Suppose that $x_k = 0$ for some $k$. In $0 = x_k = \sum_l F_{kl}\, b_l\, x_l + \delta_{ki}$ every term is non-negative, so $k \neq i$, and $x_l = 0$ for every $l$ with $F_{kl} > 0$, since $b_l > 0$. The zero entries of $\mathbf{x}$ are therefore closed under direct exchange, and by irreducibility they would include $i$, which is impossible, since $x_i = 0$ would require $\delta_{ii} = 0$. Hence $\mathbf{D}^{-1}$ is strictly positive. Finally, $\mathbf{A}^\top = (\mathbf{I} - \operatorname{diag}(\mathbf{b}))\,\mathbf{F}^\top$ has a positive diagonal factor when all $b_k < 1$, and every row of $\mathbf{F}^\top$ has a positive entry by irreducibility, so $\mathbf{A}^\top\mathbf{D}^{-1}$ is strictly positive, and the positive diagonal factor $(\mathbf{I} - \operatorname{diag}(\mathbf{b}))\mathbf{W}$ preserves this.
\end{proof}

For reciprocal $\mathbf{F}$, $\mathbf{T}$ is the matrix of total exchange areas of Noble \citep{Noble1975},
\begin{equation}
\mathbf{T} =
\begin{bmatrix}
\overline{\mathbf{SS}} & \overline{\mathbf{SG}} \\
\overline{\mathbf{GS}} & \overline{\mathbf{GG}}
\end{bmatrix} ,
\label{eq:total_exchange_blocks}
\end{equation}
which he defines as the matrix that maps the blackbody emissive powers to the absorbed powers, as evident from his equations (5) and (6); since the absorbed powers are unique for every $\mathbf{E}_\mathrm{b}$, the two matrices coincide. Noble does not assemble the four blocks into one matrix, but the symmetry relations he states for them, $\overline{\mathbf{GS}} = \overline{\mathbf{SG}}^\top$ in his equations (13b,c) and the symmetry of the groupings in his equations (13a) and (13d), amount to the symmetry of $\mathbf{T}$ for reciprocal $\mathbf{F}$. For uniform single scattering albedo $\omega_0$ and extinction coefficient, as assumed by Noble, equation \eqref{eq:total_exchange_areas} reproduces his explicit relations, equations (13a)--(13d), which evaluate the same inverse block-wise through the intermediate inverses $\mathbf{P}$ and $\mathbf{R}$ of his equations (10a) and (10b). Naraghi and Chung \citep{Naraghi1985} obtained a single relation for all four blocks by treating radiative exchange as a Markov chain, whose transition probabilities are the direct exchange areas divided by the emission capacity of the emitting zone, i.e. the exchange factors $\mathbf{F}$, with element-wise reflectivities and albedos. Their equation (22) is the transpose of equation \eqref{eq:total_exchange_areas}, and the two coincide for reciprocal $\mathbf{F}$. The practical consequence of Proposition~\ref{prop:total_exchange_positivity} is that the total exchange areas are dense whenever multiple reflection-scattering connects all elements, even when $\mathbf{F}$ is sparse, so forming them requires $\mathcal{O}\left((m+n)^2\right)$ storage. The same holds when source terms are prescribed and the zonal method is applied in its usual form, in which the energy balances are written with the total exchange areas and solved for the unknown blackbody emissive powers. The proposed formulation never forms these dense matrices: it solves $\mathbf{M}\mathbf{j} = \mathbf{h}$ once, with $\mathbf{M}$ inheriting the sparsity of $\mathbf{F}$, and obtains the absorbed powers from the sparse product $\mathbf{A}^\top\mathbf{j}$. This sparsity is a property of the discrete model rather than of the proposed formulation alone, which the total exchange areas give up: Hottel's balances, equations \eqref{eq:hottel_balance} and \eqref{eq:hottel_source}, can also be solved directly for the radiosities with the direct exchange areas, and by Theorem~\ref{thm:hottel_equivalence} the system matrix is then $\mathbf{M}\mathbf{W}$, which has the same sparsity as $\mathbf{M}$. What the proposed formulation adds are the guarantees of Sections~\ref{sec:matrix_props} and \ref{sec:physical_correct} for this sparse solve, proven for any $\mathbf{F}$ without reciprocity (Remark~\ref{rem:nonreciprocal}).

\section{Method Validation}
\label{sec:validation}

The formulation is validated at three levels, each testing a different aspect. Symbolic validation tests the formulation itself, independently of any discretization or sampling: Section~\ref{sec:symbolic_validation} recovers the textbook closed-form solutions for infinite parallel plates and concentric cylinders as exact algebraic identities, illustrates the albedo invariance of Proposition~\ref{prop:albedo_invariance} on the smallest participating system, and shows on the same system that Noble's formulation of Hottel's zonal method reproduces the solution of the proposed formulation exactly. Numerical validation, in Section~\ref{sec:numerical_validation}, tests the formulation together with ray-traced exchange factors: against the two-dimensional results of Crosbie and Schrenker \citep{Crosbie1984} for pure and partial scattering, and against Noble's formulation on a larger system, where the two agree to machine precision. The same section quantifies how the accuracy depends on the number of rays and on the sampler, measures the conditioning of the mixed-boundary system in participating media, tests the full solution against an exact symmetry identity in a three-dimensional enclosure, and compares the solution with the diffusion approximation in an optically thick problem with about $10^6$ elements. Finally, the minimum principle (Lemma~\ref{lem:min_principle}), the non-singularity of the mixed-boundary system (Theorem~\ref{thm:M_nonsingular}) and the non-negativity of its solution (Theorem~\ref{thm:physical_correctness}) have been machine-checked in the proof assistant Lean 4 \citep{deMoura2021}, using its mathematical library Mathlib \citep{Mathlib2020}, together with Theorems~\ref{thm:spectral_radius_unity}, \ref{thm:reflection_upper_bound} and \ref{thm:unconditional_conservation}, with the formal proofs available in the package repository \citep{Bielefeld2024}, so that their correctness does not rest on the written proofs alone.

Many of the proof derivations, all of the Lean 4 formalizations and the scripts that produce the figures and results of this section were developed with the assistance of Anthropic's Claude models (Claude Opus 5.5, Claude Fable 5 and Claude Fable 5.1), as detailed in the Declaration of generative AI and AI-assisted technologies.

\subsection{Symbolic Validation}
\label{sec:symbolic_validation}

A common diagnostic for any new radiative transfer formulation is whether it reproduces the classical closed-form solutions for canonical enclosures in the limit of pure surface-to-surface exchange with reflecting walls. This subsection presents a symbolic validation of the proposed formulation against the textbook formulas for two such enclosures: infinite parallel plates of equal area, and infinite concentric cylinders. The derivations are carried out symbolically in the computer algebra system SymPy \citep{SymPy2017}, with all algebraic manipulations preserved as rational expressions in the reflectivity $\rho$ and (for cylinders) the area ratio $\alpha = A_1/A_2$. The script reproducing these derivations is listed in \ref{app:symbolic_scripts} and is available as:\\ \texttt{papers/exchange-factor-formulation/symbolic/symbolic\_validation.py}\\ in the package repository \citep{Bielefeld2024}.

In both cases the medium between the surfaces is transparent, so only the surface-surface block of $\mathbf{F}$ is non-trivial, and the matrices reduce to $2 \times 2$. The reflectivity is uniform on both surfaces ($\rho_1 = \rho_2 = \rho$) and emissivity is $\varepsilon = 1 - \rho$. The blackbody emissive powers are $E_{\mathrm{b},i} = \sigma T_i^4$, and the emitted power of element $i$ is $\mathbf{e}_i = \varepsilon_i E_{\mathrm{b},i} A_i$ in Watts. The textbook reference formulas are taken from \citep{Daun2021}.

\subsubsection{Infinite Parallel Plates}
\label{sec:symbolic_plates}

For infinite parallel plates of equal area $A$, the view factors are $\mathbf{F}_{11} = \mathbf{F}_{22} = 0$ and $\mathbf{F}_{12} = \mathbf{F}_{21} = 1$. The exchange factor matrix and column-constant reflection-scattering matrix are
\begin{equation}
\mathbf{F} = \begin{bmatrix} 0 & 1 \\ 1 & 0 \end{bmatrix},
\quad
\mathbf{B} = \begin{bmatrix} \rho & \rho \\ \rho & \rho \end{bmatrix}.
\end{equation}
Applying the definitions $\mathbf{A} = \mathbf{F}\circ(\mathbf{1}-\mathbf{B})$ and $\mathbf{R} = \mathbf{F}\circ\mathbf{B}$:
\begin{equation}
\mathbf{A} = \begin{bmatrix} 0 & 1-\rho \\ 1-\rho & 0 \end{bmatrix},
\quad
\mathbf{R} = \begin{bmatrix} 0 & \rho \\ \rho & 0 \end{bmatrix}.
\end{equation}
The matrices $\mathbf{C}$ and $\mathbf{D}$ are
\begin{equation}
\mathbf{C} = \begin{bmatrix} 1 & -1 \\ -1 & 1 \end{bmatrix},
\quad
\mathbf{D} = \begin{bmatrix} 1 & -\rho \\ -\rho & 1 \end{bmatrix}.
\end{equation}
Note that $\mathbf{C}$ depends only on the geometry; the optical properties enter only through $\mathbf{D}$.

With both temperatures prescribed, the mixed-boundary matrix $\mathbf{M}$ is assembled entirely from rows of $\mathbf{D}$, and the right-hand side is $\mathbf{h} = [\mathbf{e}_1, \mathbf{e}_2]^\top = [\varepsilon E_{\mathrm{b},1} A, \varepsilon E_{\mathrm{b},2} A]^\top$. Solving $\mathbf{D}\mathbf{j} = \mathbf{h}$ yields
\begin{equation}
\mathbf{j} = \frac{\varepsilon A}{1-\rho^2} \begin{bmatrix} E_{\mathrm{b},1} + \rho E_{\mathrm{b},2} \\ E_{\mathrm{b},2} + \rho E_{\mathrm{b},1} \end{bmatrix},
\end{equation}
which is the classical radiosity result for parallel plates in power form. The absorbed power vector is $\mathbf{g}_\mathrm{a} = \mathbf{A}^\top \mathbf{j}$, and the source term on plate 1 is $q_1 = e_1 - g_\mathrm{a,1}$. After simplification:
\begin{equation}
q_1 = \frac{\sigma(T_1^4 - T_2^4) A}{1/\varepsilon + 1/\varepsilon - 1},
\label{eq:plates_textbook}
\end{equation}
which is the textbook formula for net radiative heat transfer between two infinite gray-diffuse parallel plates \citep{Daun2021}. Energy conservation $q_1 + q_2 = 0$ holds identically, in agreement with Theorem~\ref{thm:unconditional_conservation}.

The agreement is symbolic, meaning equation~\eqref{eq:plates_textbook} is recovered as an algebraic identity in $\rho$, $\varepsilon$, $\sigma$, $T_1$, $T_2$, and $A$, not as a numerical match within some tolerance.

\subsubsection{Infinite Concentric Cylinders}
\label{sec:symbolic_cylinders}

For infinite concentric cylinders (the analysis applies equally to concentric spheres) with area ratio $\alpha = A_1/A_2 < 1$, the view factors are $\mathbf{F}_{11} = 0$, $\mathbf{F}_{12} = 1$, $\mathbf{F}_{21} = \alpha$, $\mathbf{F}_{22} = 1-\alpha$. The exchange factor matrix is
\begin{equation}
\mathbf{F} = \begin{bmatrix} 0 & 1 \\ \alpha & 1-\alpha \end{bmatrix},
\end{equation}
with $\mathbf{B}$ as in section \ref{sec:symbolic_plates}. The system matrices are
\begin{equation}
\mathbf{A} = \begin{bmatrix} 0 & 1-\rho \\ \alpha(1-\rho) & (1-\alpha)(1-\rho) \end{bmatrix},
\quad
\mathbf{R} = \begin{bmatrix} 0 & \rho \\ \alpha\rho & (1-\alpha)\rho \end{bmatrix},
\end{equation}
\begin{equation}
\mathbf{C} = \begin{bmatrix} 1 & -\alpha \\ -1 & \alpha \end{bmatrix},
\quad
\mathbf{D} = \begin{bmatrix} 1 & -\alpha\rho \\ -\rho & 1-(1-\alpha)\rho \end{bmatrix}.
\end{equation}
Again, $\mathbf{C}$ is purely geometric. With both temperatures prescribed and $\mathbf{M} = \mathbf{D}$, solving $\mathbf{D}\mathbf{j} = \mathbf{h}$ and computing $q_1 = e_1 - (\mathbf{A}^\top\mathbf{j})_1$ yields
\begin{equation}
q_1 = \frac{\sigma(T_1^4 - T_2^4) A_1}{1/\varepsilon_1 + (A_1/A_2)(1/\varepsilon_2 - 1)},
\label{eq:cylinders_textbook}
\end{equation}
which is the textbook formula for net radiative heat transfer between two infinite gray-diffuse concentric cylinders \citep{Daun2021}. Again, the agreement is symbolic.

\subsubsection{Total radiant power: albedo invariance}
\label{sec:radiant_power_albedo}

A further symbolic test illustrates Proposition~\ref{prop:albedo_invariance} on the smallest possible system: an absorbing, emitting and scattering medium in radiative equilibrium with constant extinction, whose total radiant power field is independent of the single scattering albedo.

To keep the derivation short, an abstract two-zone system is used, consisting of a single surface element and a single volume element with uniform exchange factors $\mathbf{F}_{ss} = \mathbf{F}_{sg} = \mathbf{F}_{gs} = \mathbf{F}_{gg} = 1/2$, for which reciprocity requires $A=4\beta V$. Since the formulation depends on the geometry only through $\mathbf{F}$, the derivation does not require a particular geometry that realizes these exchange factors. The surface has an emissivity of unity ($\rho=0$), and the volume has a single scattering albedo of $\omega$.

For the present $2\times 2$ system ($m+n=2$) with uniform exchange factors, surface reflectivity $\rho=0$, and volume scattering albedo $\omega$, the column-constant single reflection-scattering matrix is
\begin{equation}
\mathbf{B} = \begin{bmatrix} 0 & \omega \\ 0 & \omega \end{bmatrix}.
\end{equation}
Applying the proposed formulation $\mathbf{A} = \mathbf{F} \circ (\mathbf{1}-\mathbf{B})$ and $\mathbf{R} = \mathbf{F} \circ \mathbf{B}$ from equations \eqref{eq:A} and \eqref{eq:R}, the absorption and reflection-scattering matrices are
\begin{equation}
\mathbf{A} = \begin{bmatrix} 1/2 & (1-\omega)/2 \\ 1/2 & (1-\omega)/2 \end{bmatrix},
\quad
\mathbf{R} = \begin{bmatrix} 0 & \omega/2 \\ 0 & \omega/2 \end{bmatrix}.
\end{equation}
From equation \eqref{eq:C_and_D}, the matrices $\mathbf{C}$ and $\mathbf{D}$ are
\begin{equation}
\mathbf{C} = \begin{bmatrix} 1/2 & -1/2 \\ -1/2 & 1/2 \end{bmatrix},
\quad
\mathbf{D} = \begin{bmatrix} 1 & 0 \\ -\omega/2 & 1-\omega/2 \end{bmatrix}.
\end{equation}
Notably, $\mathbf{C}$ depends only on the geometry: the optical properties enter only through $\mathbf{D}$. The radiative equilibrium condition for the volume element ($q_2 = 0$) is imposed by taking the second row of $\mathbf{C}$, and the surface emissive power $h_1$ is prescribed by taking the first row of $\mathbf{D}$, yielding
\begin{equation}
\mathbf{M} = \begin{bmatrix} 1 & 0 \\ -1/2 & 1/2 \end{bmatrix},
\quad
\mathbf{h} = \begin{bmatrix} h_1 \\ 0 \end{bmatrix}.
\end{equation}
Solving $\mathbf{M}\mathbf{j} = \mathbf{h}$ gives
\begin{equation}
\mathbf{j} = \begin{bmatrix} h_1 \\ h_1 \end{bmatrix},
\end{equation}
confirming that the total radiation field is invariant under variations in $\omega$. The element-wise energy balance is
\begin{equation}
\mathbf{C}\mathbf{j} = \begin{bmatrix} 0 \\ 0 \end{bmatrix},
\end{equation}
confirming that energy conservation is satisfied not only globally ($\mathbf{1}^\top\mathbf{C}\mathbf{j} = 0$, which holds unconditionally by Theorem~\ref{thm:unconditional_conservation}) but also element-wise in this radiative-equilibrium configuration. The emissive power vector recovered from $\mathbf{e} = \mathbf{D}\mathbf{j}$ is
\begin{equation}
\mathbf{e} = \begin{bmatrix} h_1 \\ h_1(1-\omega) \end{bmatrix},
\end{equation}
and the reflected-scattered power vector $\mathbf{r} = \mathbf{R}^\top \mathbf{j}$ is
\begin{equation}
\mathbf{r} = \begin{bmatrix} 0 \\ h_1\omega \end{bmatrix},
\end{equation}
which confirms that $\omega$ shifts the balance between the scattered power and the emissive power of the volume element linearly, as expected.

\subsubsection{Comparison to Noble's Formulation of Hottel's Zonal Method}
\label{sec:hottel_symbolic}

The equivalence established in Section~\ref{sec:hottel_equivalence} is illustrated here on the $2\times 2$ system of section \ref{sec:radiant_power_albedo}, using Noble's matrix formulation \citep{Noble1975} of Hottel's method, which includes multiple reflection-scattering. The emitted powers obtained with the proposed formulation are converted to blackbody emissive powers following Remark~\ref{rem:temperature_recovery}, inserted into Noble's relations, and the resulting source terms are compared symbolically. In section \ref{sec:hottel_numerical}, a larger system with the same boundary condition structure is compared numerically.

Noble's formulation of Hottel's method is applied. Noble's matrices (scalars) here should not be confused with the proposed matrices of this paper, even though some share the same symbols. The $2\times 2$ exchange factor matrix of section \ref{sec:radiant_power_albedo} is translated to direct exchange areas by multiplying each row with the emission capacity of its element:
\begin{equation}
	\mathbf{W}\mathbf{F} =
	\begin{bmatrix}
	\overline{\mathbf{ss}} &\overline{\mathbf{sg}}\\
	\overline{\mathbf{gs}} &\overline{\mathbf{gg}}
	\end{bmatrix} = 
	\begin{bmatrix}
	A/2 &A/2\\
	4\beta V/2 &4\beta V/2
	\end{bmatrix}
\end{equation}
This matrix must be symmetric to satisfy reciprocity, which confirms that $A=4\beta V$. Using $\omega$ and $\rho=0$ to calculate the intermediate matrices (scalars):
\begin{equation}
	\mathbf{P} = [4\beta V \mathbf{I}-\omega \overline{\mathbf{gg}}]^{-1} = \frac{1}{4\beta V - \omega 4\beta V/2 }
\end{equation}
\begin{equation}
	\mathbf{L} = \overline{\mathbf{sg}}\;\mathbf{P}\; \overline{\mathbf{gs}} = A/2\frac{1}{4\beta V - \omega 4\beta V/2 }4\beta V/2 = \frac{1}{2}\frac{A}{2 - \omega}
\end{equation}
\begin{equation}
	\mathbf{R} = [A\mathbf{I} - (\overline{\mathbf{ss}}+\omega \mathbf{L})\rho\mathbf{I}]^{-1} =
\frac{1}{A}
\end{equation}
\begin{equation}
	\mathbf{K} = 4 \beta\; \overline{\mathbf{sg}}\; \mathbf{P}\; V\mathbf{I} = 4\beta A/2 \frac{1}{4\beta V - \omega 4\beta V/2 } V = 
	\frac{A}{2 - \omega } = 2 \mathbf{L}
\end{equation}
Next, calculating the total exchange areas, according to Noble:
\begin{equation}
	\overline{\mathbf{SS}} = \varepsilon A\mathbf{I}\;\mathbf{R} (\overline{\mathbf{ss}}+\omega\mathbf{L})\varepsilon \mathbf{I} = A/2+\frac{\omega}{2}\frac{A}{2 - \omega} = \frac{A}{2-\omega}
	\label{eq:SS}
\end{equation}
\begin{equation}
	\overline{\mathbf{SG}} = \overline{\mathbf{GS}}^\top =
	(1-\omega)\varepsilon A\mathbf{I}\;\mathbf{R}\;\mathbf{K} = (1-\omega)A/(2-\omega)
\end{equation}
\begin{equation}
\begin{split}
	\overline{\mathbf{GG}} &= (1-\omega)^2 4\beta V \mathbf{I}\;\mathbf{P}\; \overline{\mathbf{gg}} + (1-\omega)^2 \mathbf{K}^\top \rho\mathbf{I}\; \mathbf{R}\;\mathbf{K}\\ &= (1-\omega)^2 4\beta V \frac{1}{4\beta V - \omega 4\beta V/2 } 4\beta V/2 \\
	&= \frac{(1-\omega)^2 4\beta V}{2 - \omega }
\end{split}
	\label{eq:GG}
\end{equation}
Next, the emitted power solution $\mathbf{e}$ of section \ref{sec:radiant_power_albedo}, calculated using the proposed exchange factor formulation, is used as an input to Noble's relations. These require the blackbody emissive powers $\mathbf{E}$ and $\mathbf{E}_\mathrm{g}$, which are obtained by dividing each emitted power by $(1-b_k)\,w_k$, following Remark~\ref{rem:temperature_recovery}, i.e. by $\varepsilon A$ for the surface and by $4\kappa V$ for the gas:
\begin{equation}
	\begin{bmatrix}
	\mathbf{E} \\
	\mathbf{E}_\mathrm{g}
	\end{bmatrix} = 
	\begin{bmatrix}
	h_1/A \\
	h_1(1-\omega)/(4\kappa V)
	\end{bmatrix} =
	\begin{bmatrix}
	h_1/A \\
	h_1/A
	\end{bmatrix},
\end{equation}
using $\varepsilon = 1$, $\kappa = (1-\omega)\beta$ and $A = 4\beta V$. Both elements have the same blackbody emissive power, as expected for a medium in radiative equilibrium inside a black enclosure at uniform temperature.

Next, according to Noble, the net source terms of the surface and volume elements are given by:
\begin{equation}
\begin{split}
	\mathbf{Q} &= \varepsilon A\; \mathbf{E}-\overline{\mathbf{SS}}\;\mathbf{E} -\overline{\mathbf{SG}}\;\mathbf{E}_\mathrm{g} \\
	\mathbf{S} &= \overline{\mathbf{GG}}\; \mathbf{E}_\mathrm{g} +\overline{\mathbf{GS}}\; \mathbf{E} -(1-\omega)4\beta V \mathbf{E}_\mathrm{g}
\end{split}
\end{equation}
inserting $\mathbf{E} = \mathbf{E}_\mathrm{g} = h_1/A$ and using $4\beta V = A$ leads to:
\begin{equation}
\begin{split}
	\mathbf{Q} &= h_1-\overline{\mathbf{SS}}\,h_1/A -\overline{\mathbf{SG}}\,h_1/A \\
	\mathbf{S} &= \overline{\mathbf{GG}}\,h_1/A +\overline{\mathbf{GS}}\,h_1/A -(1-\omega) h_1
\end{split}
\end{equation}
The total exchange areas are given by equations \eqref{eq:SS}-\eqref{eq:GG}, and using $4\beta V=A$ and inserting them yields the following expressions for the normalized net flux of each element:
\begin{equation}
\begin{split}
	\mathbf{Q}/h_1 &= 1 - \frac{1 + (1-\omega)}{2-\omega} = 0 \\
	\mathbf{S}/h_1 &= \frac{(1-\omega)^2 + (1-\omega)}{2-\omega} - (1-\omega) = 0
\end{split}
	\label{eq:hottel_agreement}
\end{equation}
Both net source terms vanish identically for all $\omega$, as required by radiative equilibrium, in agreement with the element-wise balance $\mathbf{C}\mathbf{j} = \mathbf{0}$ of the proposed formulation in section \ref{sec:radiant_power_albedo}. Noble's formulation of Hottel's method therefore reproduces the solution of the proposed formulation exactly, as guaranteed by Theorem~\ref{thm:hottel_equivalence}. Note that Noble's $\mathbf{S}$ is defined as absorbed minus emitted power, i.e. with the opposite sign of the source term $q$ of the present paper, which is immaterial here since both vanish.

\subsection{Numerical Validation}
\label{sec:numerical_validation}

For numerical validation, a radiative transfer code, written by the author with support from artificial intelligence (Anthropic's Claude models), in the Julia programming language \citep{Bezanson2017}, was used. 

For two-dimensional problems including participating media, this code ray traces in parallel using the central processing unit (CPU), allowing for higher numerical precision than typical graphical processing unit (GPU) ray tracing, and greater data structure flexibility. For the present analysis, the code assumes a uniform extinction coefficient $\beta$, and ray traces on a coarse mesh, then maps the point of first interaction to a finer mesh, using a grid mapping to find the absorber indices, which allows for efficient ray tracing on arbitrarily fine meshes on the CPU. The CPU implementation offers generality, data handling flexibility, code simplicity and code maintainability.

Importantly, the rays are traced only until the point of first interaction, as required by the proposed formulation. Unless stated otherwise, the emission points and directions of the rays are sampled from randomized Sobol sequences \citep{Sobol1967, JoeKuo2008}. Each emitter uses its own Sobol sequence, with one dimension per random number that a ray consumes, randomized by a random digital shift: every coordinate is combined by a bitwise exclusive-or with a random mask, obtained from a SplitMix64 hash of the seed, the emitter and the coordinate. The digital shift preserves the low-discrepancy structure of the sequence while making every estimate unbiased, and different seeds give independent randomizations. Each row of a ray-traced exchange factor matrix is normalized by the number of rays recorded for its emitter, so that $\mathbf{F}$ is row-stochastic even when individual rays are lost in the tracing. Apart from this normalization every ray-traced exchange factor matrix in this section is used as traced, without smoothing, unless stated otherwise.

The code also features a Julia adaptation, by permission from the authors, of the MATLAB code by Jacob A. Kerkhoff and Michael J. Wagner, published in their GitHub repository \citep{Wagner2021}, during their work on a paper on solar cavity receivers \citep{Kerkhoff2021}, which implements the three-dimensional analytical view factor method of Narayanaswamy \citep{Narayanaswamy2015}. For the present work, the adaptation was made generic in the floating-point type, so that the view factors and the subsequent linear solve can be carried out in any precision. To this end, the Chebyshev approximation of the Clausen function given by Narayanaswamy was replaced by its power series in the Bernoulli numbers \citep[Eq.~27.8.2]{AbramowitzStegun1964}, which converges for $|\theta| < 2\pi$ and is applied after reducing the argument to $|\theta| \leq \pi$, summed until the terms fall below the machine precision of the working type, with the Bernoulli numbers obtained exactly from the integer tangent numbers by the recurrence of Brent and Harvey \citep{BrentHarvey2013}.

The code written for validation is available as a registered package for radiative transfer calculations in the Julia programming language \citep{Bielefeld2024}. The package is under active development, with documentation and methodology details being continually rolled out at \url{https://gert.net}. Every figure and every number reported in this section is produced by a script in the package repository, in the folder \texttt{papers/exchange-factor-formulation/figures}, with one script per figure, so that all results can be reproduced with the published package.

The hardware used for this validation is a 64-bit PC with an Intel\textsuperscript{\tiny\textregistered} Core$^\mathrm{TM}$ i7-14700KF, 3400 MHz, 20 Cores, 28 logical processors CPU and 64 GB of RAM, running Julia 1.13.1. The Sobol results are reproducible across Julia versions, since the randomization is generated by the package itself; the pseudorandom results depend on Julia's default random number stream, which is stable only within a minor release.

\subsubsection{Comparison to Crosbie and Schrenker}
\label{sec:Crosbie_Schrenker}

The two-dimensional results of Crosbie and Schrenker \citep{Crosbie1984} are used for validation. They consider a square medium exposed to diffuse radiation of intensity $i_\mathrm{s}$ from one black surface, with the remaining surfaces black and cold, an isotropically scattering medium of single scattering albedo $\omega$, and the thermal emission of the medium neglected. In the proposed formulation, the emission of the medium is neglected by prescribing its temperature as zero Kelvin, which sets the emitted powers of the gas elements to zero, while $\mathbf{R}$ from equation \eqref{eq:R} still scatters the incident radiation. For $\omega = 1$ this coincides with radiative equilibrium. With the emission neglected, the outgoing power of gas volume element $k$ is $j_k = \omega g_k$, so its intensity relative to that of the illuminating diffuse surface is $i_{\mathrm{g},k}/i_\mathrm{s} = j_k/(w_k \sigma T_\mathrm{hot}^4) = S/i_\mathrm{s}$. Crosbie and Schrenker tabulate the source function with the albedo taken out, $\phi = S/(\omega\, i_\mathrm{S})$, which is therefore the incident power per unit emission capacity,
\begin{equation}
	\left(\frac{S}{\omega}\right)_k = \frac{(\mathbf{F}^\top \mathbf{j})_k}{w_k\, \sigma T_\mathrm{hot}^4},
	\label{eq:crosbie_compare}
\end{equation}
where $(\mathbf{F}^\top \mathbf{j})_k$ is the radiant power arriving at gas element $k$ at its first interaction and $w_k = 4\beta V_k$.

Six cases are compared on a $21\times 21$ mesh, listed in Table \ref{tab:table1}: an albedo series at $\tau_L = 1$ and an optical-thickness series at $\omega = 1$. Since $\mathbf{F}$ depends only on the extinction coefficient, one exchange factor matrix is traced per optical thickness and reused for every albedo. Each element emits $2^{16}$ rays, increased in proportion to $1/\tau_L$ for $\tau_L < 1$, since a ray interacts within an optically thin element with a probability proportional to its optical thickness. 

A computed element value is an average over the element, whereas the reference is tabulated along the centre line. Since the profile across the medium peaks at the centre line, the average over the middle column lies below the centre-line value, by $h^2/24$ times the magnitude of the curvature across the centre line. This curvature is estimated from the second difference of the three middle columns and the offset is added back, while the reference is averaged over each row of elements. For the averaging, the tabulated reference is interpolated with the Julia package ConvolutionInterpolations.jl \citep{Bielefeld2025}, using its kernel \texttt{:b13}. On nonuniform grids such as the tabulated optical depths, the kernel's polynomial weights are adjusted to the local grid spacing, so that it retains a high order of convergence. Figure \ref{fig2} shows the results. The root mean square deviation from the reference is at most $2.2\times 10^{-3}$ in all six cases, and the relative energy conservation error $\mathbf{1}^\top\mathbf{q}/\mathbf{1}^\top\mathbf{j}$, the net source summed over all elements relative to the total radiant power, is below $10^{-16}$. This error can fall below the machine precision, since the round-off errors of the individual source terms partially cancel in the sum.

\begin{table} [ht]
\centering
\caption{The cases from Crosbie and Schrenker \citep{Crosbie1984} used for validation: square media with diffuse radiation incident from one black wall, all walls black, optical thickness $\tau_L$, single scattering albedo $\omega$ and $N_\mathrm{b}/N$ rays traced per element. The last column gives the root mean square deviation of the computed centre-line $S/\omega$ from the row-averaged reference, on a $21\times 21$ mesh with no smoothing applied.}
\label{tab:table1}
\begin{tabular}{llll}
$\tau_L$ & $\omega$ & $N_\mathrm{b}/N$ & RMS deviation \\ \hline
1    & 1.0 & $2^{16}$ & $2.2\times 10^{-3}$ \\
1    & 0.9 & $2^{16}$ & $2.1\times 10^{-3}$ \\
1    & 0.5 & $2^{16}$ & $2.0\times 10^{-3}$ \\
0.25 & 1.0 & $2^{18}$ & $2.0\times 10^{-3}$ \\
4    & 1.0 & $2^{16}$ & $1.0\times 10^{-3}$ \\
10   & 1.0 & $2^{16}$ & $2.2\times 10^{-3}$ \\ \hline
\end{tabular}
\end{table}

\begin{figure}
\centering
\includegraphics[width=\linewidth]{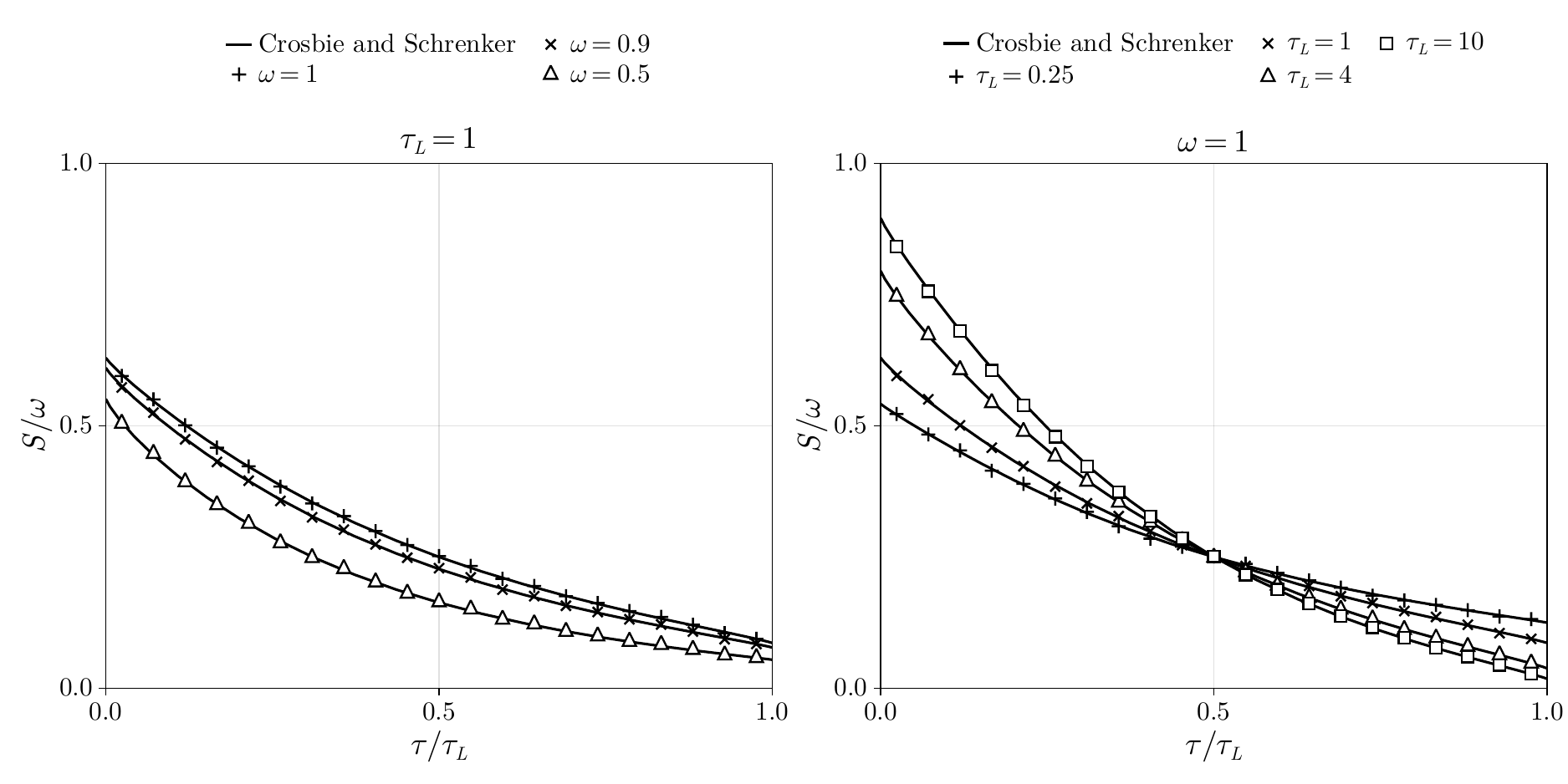}
\caption{Centre-line mean incident intensity $S/\omega$, perpendicular to the illuminated wall, as a function of the relative optical depth $\tau/\tau_L$, in square media with diffuse radiation incident from one black wall. Left: albedo series at $\tau_L = 1$. Right: optical-thickness series at $\omega = 1$. Markers show the proposed formulation on a $21\times 21$ mesh with no smoothing applied, and lines the solutions of Crosbie and Schrenker \citep{Crosbie1984}.}\label{fig2}
\end{figure}

\subsubsection{Numerical Agreement with Noble's Formulation of Hottel's Method}
\label{sec:hottel_numerical}

This section confirms the agreement of section \ref{sec:hottel_symbolic} numerically on a larger system with the same boundary condition structure: a $21\times 21$ unit square with extinction $\beta=1$, black walls with the bottom wall at 1000 Kelvin and the others at zero Kelvin, and the gas in radiative equilibrium, for single scattering albedos $\omega = 0, 0.1, \ldots, 0.9$. For each albedo, the proposed formulation is solved with a single linear solve, the blackbody emissive powers are recovered from its solution following Remark~\ref{rem:temperature_recovery}, and Noble's net source terms $\mathbf{Q}$ and $\mathbf{S}$ are computed from them with the relations of section \ref{sec:hottel_symbolic}, applying the intermediate operators $\mathbf{P}$ and $\mathbf{R}$ by linear solves rather than explicit inverses. At $\omega = 1$ the gas does not emit, so its blackbody emissive power is undefined in Noble's variables, and this case is covered by the symbolic comparison.

Noble's explicit relations use the direct exchange areas as a symmetric matrix, which an exchange factor matrix estimated by ray tracing provides only up to its statistical noise. For this comparison, the exchange factors traced with $N_\mathrm{b}/N=2^{16}$ rays per element were therefore made exactly reciprocal and closed by smoothing \citep{Daun2005}, as implemented in the package used here; everywhere else in this paper the exchange factors are used as traced. With this exchange factor matrix, the largest difference between Noble's source terms and those of the proposed formulation is below $5\times 10^{-17}$ for the surfaces and below $5\times 10^{-18}$ for the gas, relative to the power emitted by the walls, for every albedo. The two formulations therefore agree to machine precision, as guaranteed by Theorem~\ref{thm:hottel_equivalence}. They differ in cost: Noble's total exchange areas require linear solves with as many right-hand sides as there are elements and produce dense matrices, whereas the proposed formulation obtains the same solution from a single linear solve with one right-hand side, whose system matrix inherits the sparsity of $\mathbf{F}$.

The same script also solves Hottel's balances, equations \eqref{eq:hottel_balance} and \eqref{eq:hottel_source}, directly for the radiosities with the exchange factors as traced, without smoothing, at $\omega = 0.5$. With the direct exchange areas defined from the emitter, $\overline{s_i s_k} = w_i F_{ik}$, the source terms differ from those of the proposed formulation by less than $3\times 10^{-17}$ relative to the power emitted by the walls, and the relative energy conservation error is below $10^{-16}$, as Theorem~\ref{thm:hottel_equivalence} and Theorem~\ref{thm:unconditional_conservation} require.

\subsubsection{Convergence and Uncertainty}
\label{sec:convergence}

Since the proposed formulation is exact for a given $\mathbf{F}$, the accuracy of its solution is limited by that of the exchange factors, which for participating media are estimated by ray tracing. To quantify this, the case $\tau_L = 1$, $\omega = 1$ of section \ref{sec:Crosbie_Schrenker} is solved on an $11\times 11$ mesh, equivalently as an absorbing-emitting medium in radiative equilibrium (Proposition~\ref{prop:albedo_invariance}), with $2^{10}$ to $2^{20}$ rays per element, and the centre-line source function is compared with Crosbie and Schrenker \citep{Crosbie1984} as in section \ref{sec:Crosbie_Schrenker}. Two samplers are compared: pseudorandom numbers, i.e. classical Monte Carlo, and randomized Sobol sequences \citep{Sobol1967, JoeKuo2008}. For each sampler and number of rays, ten independent replicates are traced, with independent random streams or independent randomizations of the Sobol sequence, and the uncertainty of the solution is quantified directly by the standard deviation of its error across the replicates.

Figure \ref{fig4} shows the result. For both samplers the root mean square deviation decreases in proportion to $N_\mathrm{b}^{-1/2}$, with fitted slopes of $-0.51$ (pseudorandom) and $-0.53$ (Sobol) over three decades of ray counts, and shows no sign of a floor from the formulation itself. Sobol sampling does not improve the rate, consistent with integrands whose discontinuities, at visibility and element boundaries, are not aligned with the coordinate axes, which limits the gain from low-discrepancy sampling, but it lowers the deviation by a factor of 1.4 to 2.3 at every number of rays and narrows the spread across replicates, by a factor of about five at $2^{20}$ rays per element. The same accuracy is therefore reached with roughly two to five times fewer rays. At $2^{20}$ rays per element the Sobol deviation is $4.3\times 10^{-4}$, which is the precision of the tabulated reference itself; below this level the deviation can no longer be attributed to the ray count alone, since the discretization error of the $11\times 11$ mesh may begin to contribute.

\begin{figure}
\centering
\includegraphics[width=0.6\linewidth]{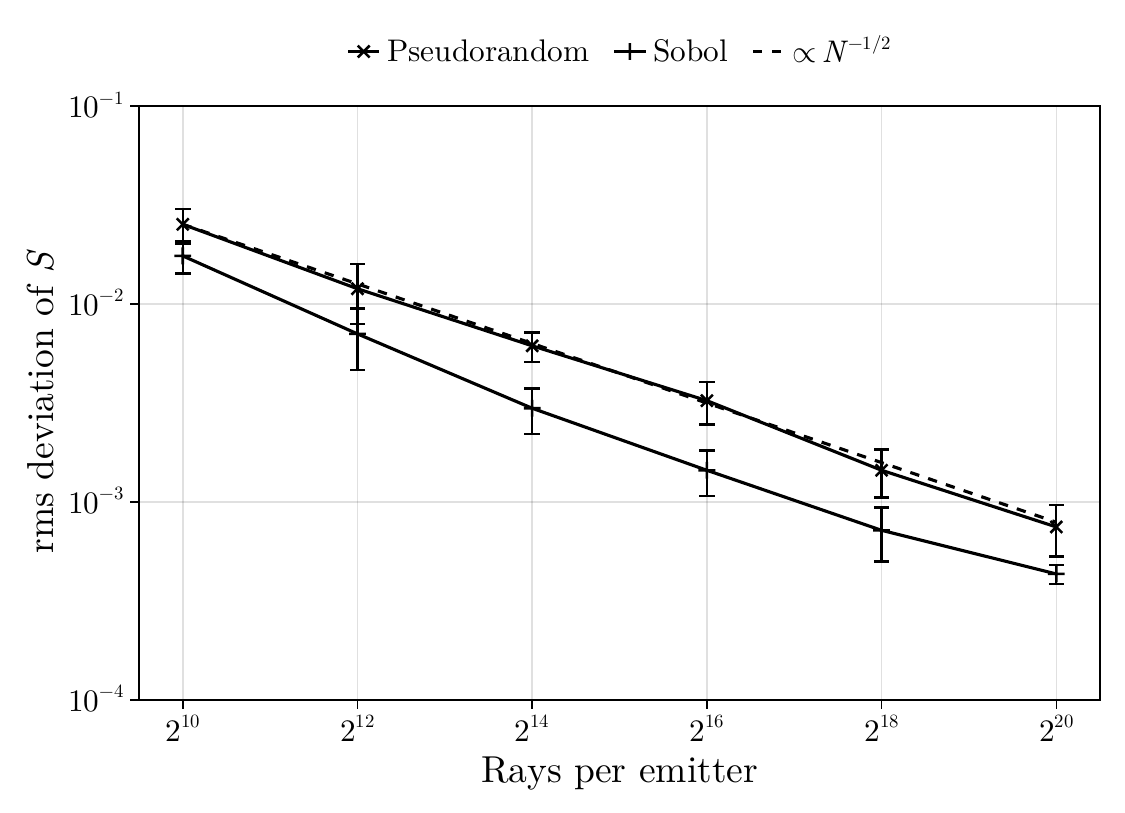}
\caption{Root mean square deviation of the centre-line source function $S$ from Crosbie and Schrenker \citep{Crosbie1984} ($\tau_L = 1$, $\omega = 1$) on an $11\times 11$ mesh, against the number of rays per element, for pseudorandom and Sobol sampling, with no smoothing applied. Markers show the mean over ten independent replicates and error bars one standard deviation across them; the dashed line indicates the rate $N_\mathrm{b}^{-1/2}$.}\label{fig4}
\end{figure}

\subsubsection{Conditioning}
\label{sec:conditioning}

Theorem~\ref{thm:M_nonsingular} guarantees that the mixed-boundary system has a unique solution; its condition number determines how many digits of precision the linear solve can lose. When all emitted powers are prescribed, $\mathbf{M} = \mathbf{D}$, and Corollary~\ref{cor:D_conditioning} bounds the condition number by $(1+b_{\max})/(1-b_{\max})$, independently of the number of elements, the geometry and the optical thickness, with equality for uniform reflectivity in a transparent enclosure with flat walls. A reflectivity of $0.99$, for example, gives a condition number of at most $199$. This case therefore requires no numerical study.

When the gas is in radiative equilibrium, the gas rows of $\mathbf{M}$ are taken from $\mathbf{C} = \mathbf{I} - \mathbf{F}^\top$ and no such a priori bound is available. Figure \ref{fig6} shows the condition number $\operatorname{cond}_1(\mathbf{M})$ for a unit square with black walls at prescribed temperatures and the gas in radiative equilibrium, for optical thicknesses $\tau_L = 0.1$ to $100$ and meshes of $5\times 5$ to $40\times 40$ elements, with exchange factors traced with $2^{12}$ rays per element. With black walls and the gas rows taken from $\mathbf{C}$, $\mathbf{M}$ does not depend on the single scattering albedo, by equation \eqref{eq:M_masked}. For $\tau_L \leq 1$ the condition number is below $4$ and independent of the mesh. In optically thick media it increases with the optical thickness, and with refinement only while the elements remain optically thick: at $\tau_L = 10$ it levels off near $59$ once the optical thickness of an element, $\tau_L/n_x$ with $n_x$ elements per side, falls below unity, while at $\tau_L = 100$, where the elements of the finest mesh still have an optical thickness of $2.5$, it reaches $2.3\times 10^{3}$. This behaviour reflects the diffusive character of radiative transfer in optically thick media rather than a limitation of the formulation, and even the largest value leaves more than twelve of the sixteen significant digits of double precision.

\begin{figure}
\centering
\includegraphics[width=0.6\linewidth]{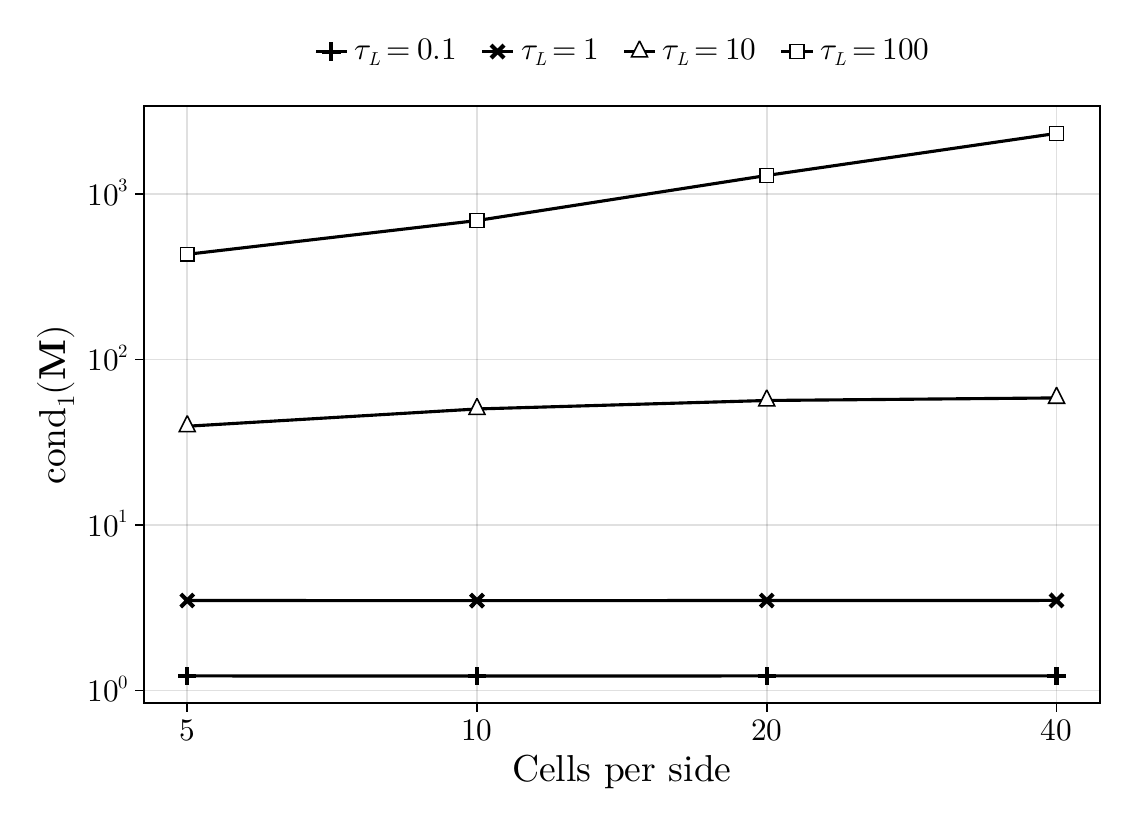}
\caption{Condition number in the 1-norm of the mixed-boundary matrix $\mathbf{M}$ for a unit square with black walls at prescribed temperatures and the gas in radiative equilibrium, against the number of elements per side, for four optical thicknesses $\tau_L$. The exchange factors are traced with $2^{12}$ rays per element, with no smoothing applied.}\label{fig6}
\end{figure}

\subsubsection{Symmetry Validation in a Transparent Three-Dimensional Enclosure}
\label{sec:3d}

Figure \ref{fig10} shows the proposed formulation applied to a transparent three-dimensional enclosure: a black sphere approximated by an icosphere, i.e. a regular icosahedron whose triangles are recursively subdivided and projected onto the unit sphere. Since the medium is transparent, only the surface block $\mathbf{F}_\mathrm{ss}$ of the exchange factor matrix is non-zero, and the formulation applies to it without modification. Each triangle is one element. The six highest triangles form a hot cap with a prescribed temperature of $T_\mathrm{hot} = 1000$ Kelvin, the six lowest, which are their images under the antipodal map, form a cold cap at $T_\mathrm{cold} = 0$ Kelvin, and all remaining triangles have a prescribed source of zero, meaning they are in radiative equilibrium.

This configuration admits an exact identity, which tests the computed solution itself rather than its agreement with a continuous reference. With black surfaces, $\mathbf{D} = \mathbf{I}$, the radiant power of element $k$ equals its emitted power, $j_k = A_k E_k$ with $E_k = \sigma T_k^4$, and the equilibrium condition $q_k = (\mathbf{C}\mathbf{j})_k = 0$ reads $A_k E_k = \sum_l F_{lk} A_l E_l$. The exact view factors of the closed icosphere are reciprocal, $A_l F_{lk} = A_k F_{kl}$, so
\begin{equation}
E_k = \sum_l F_{kl}\, E_l
\label{eq:mean_value}
\end{equation}
for every triangle $k$ in radiative equilibrium. Since the view factors in each row are non-negative and sum to unity, equation \eqref{eq:mean_value} expresses $E_k$ as a convex combination of the blackbody emissive powers of the triangles it sees, so the uniform field satisfies it. The identity then follows from symmetry and linearity, using that each problem below has a unique solution by Theorem~\ref{thm:M_nonsingular}, since the caps have prescribed temperatures and absorb radiation. Consider the two problems with the hot cap on top and on the bottom, respectively. The icosphere is centrally symmetric, so the antipodal map $k \mapsto -k$, which maps every triangle onto the triangle diametrically opposite to it, leaves the view factors unchanged, $F_{-k,-l} = F_{kl}$, and interchanges the two caps. The solution of the second problem is therefore the reflection $E_{-k}$ of the solution $E_k$ of the first. Since equation \eqref{eq:mean_value} is linear, the sum $E_k + E_{-k}$ solves the problem with both caps at $E_\mathrm{hot} + E_\mathrm{cold}$, whose solution is the uniform field $E_\mathrm{hot} + E_\mathrm{cold}$.

Hence $E_k + E_{-k} = E_\mathrm{hot} + E_\mathrm{cold}$ for every triangle $k$, i.e.
\begin{equation}
T_k^4 + T_{-k}^4 = T_\mathrm{hot}^4 + T_\mathrm{cold}^4 .
\label{eq:antipodal_identity}
\end{equation}
For the cap triangles this holds by construction, and for the triangles in radiative equilibrium it is a consequence of the formulation. The identity is a property of the discrete system rather than of the continuous sphere: it holds exactly at every level of subdivision, so the zoning introduces no error into it and no mesh convergence is required. The only possible sources of deviation are the floating-point evaluation of $\mathbf{F}$ and of the linear solve.

The view factors are computed analytically with the expression of Narayanaswamy \citep{Narayanaswamy2015}, with no smoothing applied. Since the Clausen function in this expression is evaluated here to working precision rather than by a fixed Chebyshev approximation, the view factors and the linear solve can be carried out in any floating-point precision. Figure \ref{fig10} shows the largest deviation from equation \eqref{eq:antipodal_identity}, relative to $T_\mathrm{hot}^4$, for icospheres of 80 and 320 triangles, in double precision (Float64) and in 128-bit arithmetic, using Julia's arbitrary-precision type BigFloat with a 128-bit significand. In double precision the deviation lies above machine precision and grows with refinement. In 128-bit arithmetic it falls at both levels by about the ratio of the two machine precisions. The deviation in double precision is therefore round-off, not an error of the formulation. The same comparison shows that the evaluation of the view factors introduces no error beyond round-off, since a fixed-precision approximation would set a floor independent of the working precision. The individual view factors are not tested, however, since the identity depends on $\mathbf{F}$ only through its closure, reciprocity and antipodal symmetry.

\begin{figure}
\centering
\includegraphics[width=\linewidth]{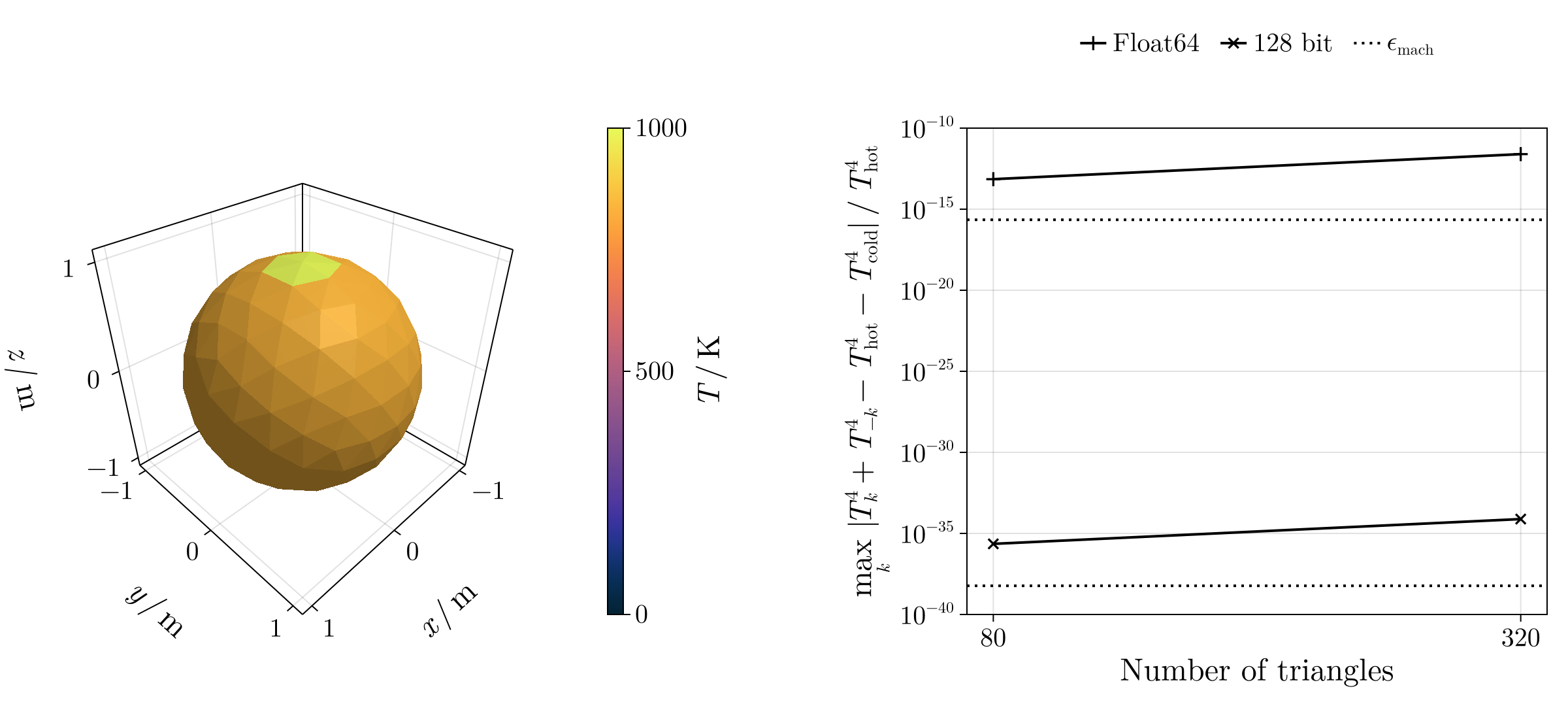}
\caption{Symmetry validation in a black icosphere enclosure with a hot cap at 1000 Kelvin on top, an antipodal cold cap at zero Kelvin, and all other triangles in radiative equilibrium. Left: temperature field on the icosphere of 320 triangles, each triangle coloured by its temperature with no interpolation. Right: largest deviation from the antipodal identity, equation \eqref{eq:antipodal_identity}, relative to $T_\mathrm{hot}^4$, against the number of triangles, in double precision (Float64) and 128-bit arithmetic, with the machine precision $\epsilon_\mathrm{mach}$ of each for reference. The view factors were computed analytically with the expression of Narayanaswamy \citep{Narayanaswamy2015}, with no smoothing applied.}\label{fig10}
\end{figure}

\subsubsection{Large-Scale Application: The Diffusion Limit}
\label{sec:large_scale}

To demonstrate the problem size that the sparsity of the formulation makes accessible, it is applied to an optically thick problem with a reference solution: a two-dimensional rectangle of 1000 m by 1 m with $\beta = 100$, filled with an absorbing-emitting, non-scattering medium in radiative equilibrium and bounded by black walls, of which one long wall is at 1000 Kelvin and the others at zero Kelvin. Far from the short walls, the problem approximates radiative transfer between infinite parallel plates at an optical thickness $\tau_L = 100$, for which the diffusion approximation with temperature-jump boundary conditions is accurate. Following Howell et al. \citep{Daun2021}, it gives
\begin{equation}
\begin{split}
	q_\mathrm{z} &= \frac{E_\mathrm{bw1}-E_\mathrm{bw2}}{3\beta D/4+1/\varepsilon_\mathrm{w1}+1/\varepsilon_\mathrm{w2}-1}\\
	E_\mathrm{b1} &= E_\mathrm{bw1} + q_\mathrm{z}(1/2 - 1/\varepsilon_\mathrm{w1})\\
	E_\mathrm{b}(z) &= E_\mathrm{b1} - (3\beta z/4) q_\mathrm{z}
\end{split}
\label{eq:diffusion_approx}
\end{equation}
where $E_\mathrm{bw1}$ and $E_\mathrm{bw2}$ are the area-specific emissive powers of the plates, $D$ is the separation distance, $q_\mathrm{z}$ is the area-specific energy transfer rate, $E_\mathrm{b1}$ is the area-specific emissive power of the medium at the surface of the first wall and $E_\mathrm{b}(z)$ is the linearly changing area-specific emissive power of the medium, as a function of position. With $E_\mathrm{bw2} = 0$ and $\varepsilon_\mathrm{w1}=\varepsilon_\mathrm{w2}=1$, the profile normalized by the hot wall simplifies to
\begin{equation}
	\frac{E_\mathrm{b}(z)}{E_\mathrm{bw1}} = \left(\frac{T(z)}{T_\mathrm{hot}}\right)^4 = 1 - \frac{3\beta z/4 + 1/2}{3\beta D/4+1},
	\label{eq:diffusion_simple}
\end{equation}
a straight line from approximately $(0,1)$ to $(1,0)$ in the relative optical depth $\tau/\tau_L = z/D$.

The domain is divided into $1001 \times 1001$ elements, which with the boundary elements gives $N = 1{,}006{,}005$ elements in total. For both samplers of Section~\ref{sec:convergence}, $2^{10}$ rays were traced per element, about $1.03\times 10^9$ in total, and the exchange factors were used without smoothing. Table~\ref{tab:large_scale} lists the results. Since rays attenuate within a short distance in this optically thick medium, each element exchanges directly with about 67 others on average, so that only $6.7\times 10^{-5}$ of the $N^2$ entries of $\mathbf{F}$ are non-zero. Stored densely in double precision, a matrix of this dimension would require 8.1~TB, whereas the proposed formulation keeps $\mathbf{F}$ sparse throughout, in about 1.1~GB. With black walls and a non-scattering medium, $b_k = 0$ for every element, so Noble's total exchange areas, equation \eqref{eq:total_exchange_areas}, reduce to the direct exchange areas $\mathbf{F}^\top\mathbf{W}$ and are equally sparse. Since the walls are black and the gas is in radiative equilibrium, however, $\mathbf{M}$ does not depend on the single scattering albedo of the gas by Proposition \ref{prop:albedo_invariance}, so the same linear system and the same solution hold for every $\omega$ at the same extinction coefficient. For $\omega>0$, the gas-gas block of the total exchange areas is strictly positive, by the argument of Proposition~\ref{prop:total_exchange_positivity} applied to the gas elements, which have $b_k = \omega > 0$ and are strongly connected through direct exchange between neighbouring elements, and storing it alone would require 8.0~TB.

The single linear system $\mathbf{M}\mathbf{j} = \mathbf{h}$ was solved without preconditioning by restarted GMRES with a Krylov subspace of dimension 50, as implemented in Krylov.jl \citep{Montoison2023}, to a relative residual of $10^{-12}$, which took about 900 iterations. Including the assembly of $\mathbf{M}$ and the recovery of temperatures, the solve took about one minute on the hardware of Section~\ref{sec:numerical_validation}. The relative energy conservation error, as defined in Section~\ref{sec:Crosbie_Schrenker}, is below $10^{-18}$ for both samplers, consistent with round-off errors that partially cancel over about $10^6$ source terms. By Theorem~\ref{thm:unconditional_conservation}, conservation holds for any $\mathbf{j}$, independently of how far the iterative solve is converged, whereas the accuracy of $\mathbf{j}$ itself is set by the residual tolerance. The smallest entry of the iterative solution $\mathbf{j}$ is exactly zero for both samplers, at the cold black walls, whose rows of $\mathbf{M}$ are rows of the identity with $h_k = 0$, and no entry is negative, although Theorem~\ref{thm:physical_correctness} guarantees non-negativity only for the exact solution.

\begin{table}[ht]
\centering
\caption{Large-scale diffusion limit: a 1000 m by 1 m rectangle with $\beta = 100$, divided into $1001\times 1001$ elements ($N = 1{,}006{,}005$ including the boundary elements), with $2^{10}$ rays traced per element and no smoothing applied. The linear system is solved by GMRES(50) to a relative residual of $10^{-12}$. The root mean square deviation of the centre-line $(T/T_\mathrm{hot})^4$ from equation \eqref{eq:diffusion_simple} is taken over all 1001 elements across the gap.}
\label{tab:large_scale}
\begin{tabular}{lll}
 & Pseudorandom & Sobol \\ \hline
Non-zero entries of $\mathbf{F}$ & $6.74\times 10^{7}$ & $6.76\times 10^{7}$ \\
Tracing / s                      & 24.8                 & 29.3 \\
GMRES iterations                 & 927                  & 909 \\
Solving / s                      & 59.2                 & 62.3 \\
RMS deviation                    & $3.1\times 10^{-2}$  & $1.6\times 10^{-2}$ \\ \hline
\end{tabular}
\end{table}

Figure~\ref{fig:large_scale} shows the centre-line profile across the gap for both samplers, with every element plotted. Both follow the diffusion approximation, with a scatter that is statistical rather than systematic: Sobol sampling lowers the root mean square deviation by about a factor of two, consistent with the constant-factor gain found in Section~\ref{sec:convergence}, at a small increase in tracing time. At $2^{10}$ rays per element, however, the noise in the exchange factors remains clearly visible, and since the deviation decreases only as $N_\mathrm{b}^{-1/2}$ in the number of rays, a tenfold reduction by ray count alone would require about $10^{11}$ rays. At this scale, the accuracy of the solution is therefore limited by the exchange factors rather than by the formulation, which motivates the variance reduction discussed in Section~\ref{sec:limitations}.

\begin{figure}
\centering
\includegraphics[width=\linewidth]{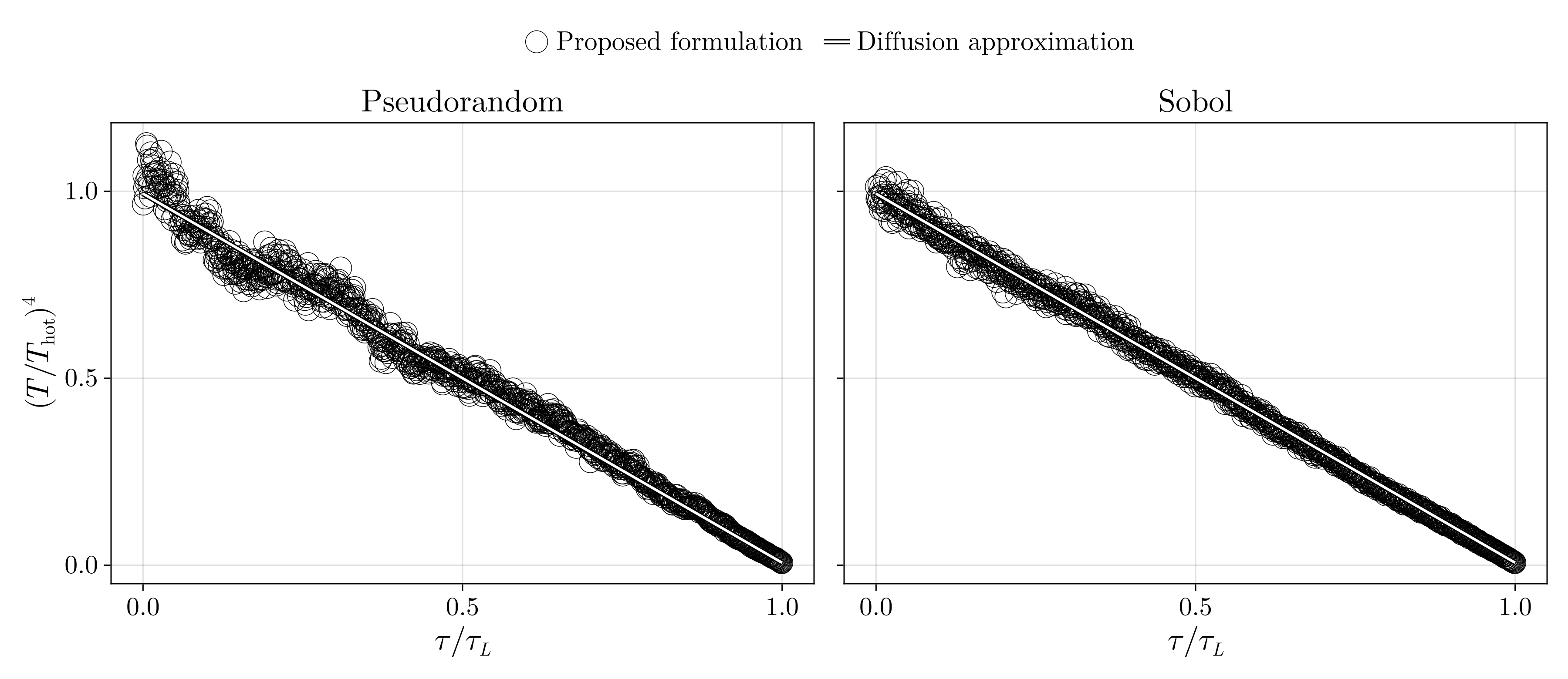}
\caption{Normalized emissive power $(T/T_\mathrm{hot})^4$ along the centre line across the gap of a 1000 m by 1 m rectangle (approximately one-dimensional), divided into $1001\times 1001$ elements, against the relative optical depth from the hot wall. The medium is absorbing-emitting and non-scattering with $\beta = 100$ ($\tau_L = 100$), in radiative equilibrium, with one long wall hot and all walls black. Circles show the proposed formulation at every element across the gap, with $\mathbf{F}$ traced with $2^{10}$ pseudorandom (left) or Sobol-sampled (right) rays per element and no smoothing applied, and the line shows the diffusion approximation between infinite parallel plates \citep{Daun2021}.}\label{fig:large_scale}
\end{figure}

\section{Discussion}
\label{sec:discussion_overall}

\subsection{Scalability Assessment}
\label{sec:scalability}

The proposed formulation requires a single linear solve $\mathbf{M}\mathbf{j} = \mathbf{h}$ to obtain the radiative transfer solution. The matrices $\mathbf{A}$, $\mathbf{R}$, $\mathbf{C}$, and $\mathbf{D}$ are constructed by Hadamard products and additions involving $\mathbf{F}$ and $\mathbf{B}$, which preserve the sparsity pattern of $\mathbf{F}$. Therefore, the sparsity of $\mathbf{F}$ propagates through the entire formulation to $\mathbf{M}$.

For low-extinction problems or problems with global radiative connectivity, $\mathbf{F}$ is dense and the dimension $m+n$ determines the memory and computational footprint. For typical desktop computing hardware, this limits applicability to medium-scale problems. For high-performance computing clusters, with greater memory availability and the option of using storage disk space to supplement working memory, larger dense problems become tractable, with the trade-off of increased time demand.

For high-extinction problems, $\mathbf{F}$ is naturally sparse: rays attenuate quickly and only nearby elements have non-negligible exchange. The proposed formulation then operates on sparse matrices throughout, enabling sparse direct or iterative linear solvers to be used for $\mathbf{M}\mathbf{j} = \mathbf{h}$. In this regime the method scales to large problems, with the achievable size determined primarily by the sparsity pattern of $\mathbf{F}$ rather than by its nominal dimension.

\subsection{Computational Complexity}

The computational complexity of the proposed formulation is dominated by the single linear solve $\mathbf{M}\mathbf{j} = \mathbf{h}$. For a dense system of dimension $m+n$, this requires one LU factorization followed by one back-substitution, totaling approximately $\frac{2}{3}(m+n)^3 + \mathcal{O}\left((m+n)^2\right)$ floating-point operations. The construction of $\mathbf{A}$, $\mathbf{R}$, $\mathbf{C}$, and $\mathbf{D}$ from $\mathbf{F}$ and $\mathbf{B}$ requires only Hadamard products and additions, costing $O((m+n)^2)$ operations and contributing negligibly to the total. Once the LU factorization of $\mathbf{M}$ has been computed, subsequent solves with different right-hand sides $\mathbf{h}$ cost only $\mathcal{O}\left((m+n)^2\right)$ operations per solve. This is a standard property of LU factorization rather than specific to the proposed formulation, but it makes parameter studies over boundary conditions at fixed $\mathbf{M}$ particularly inexpensive.

A useful structural observation is that the matrix $\mathbf{C} = \mathbf{I} - \mathbf{A}^\top - \mathbf{R}^\top = \mathbf{I} - \mathbf{F}^\top$ depends only on the geometry through $\mathbf{F}$, not on the optical properties of the surfaces or medium. Optical properties enter the linear system only through $\mathbf{D} = \mathbf{I} - \mathbf{R}^\top$. For parameter studies in which the geometry is held fixed but the optical properties are varied, the geometric matrix $\mathbf{C}$ need only be assembled once.

In the absence of reflection-scattering ($\mathbf{B} = \mathbf{0}$), $\mathbf{R} = \mathbf{0}$ and $\mathbf{D} = \mathbf{I}$. When all temperatures are known so that $\mathbf{M} = \mathbf{D} = \mathbf{I}$, the linear solve reduces to $\mathbf{j} = \mathbf{h}$ with no computational work beyond the ray-traced exchange factors. This makes the absorption-only case ideal for dynamic time-stepping simulations in which the geometry and extinction evolve slowly and the radiation field must be updated frequently, as long as photon time-of-flight effects can be neglected.

\subsection{Comparison to Existing Methods}

The proposed formulation is Hottel's zonal method \citep{Hottel1958, Hottel1967}, which partitions the radiative transfer problem into surface and gas elements with exchange quantities between them and includes reflection-scattering through Noble's matrix formulation \citep{Noble1975}, written in power variables. By Theorem~\ref{thm:hottel_equivalence}, the two solve the same discrete model for every $\mathbf{F}$, which Sections \ref{sec:hottel_symbolic} and \ref{sec:hottel_numerical} confirm symbolically and numerically to machine precision. Since Larsen and Howell related their exchange factors to the direct-exchange areas by explicit transformations, and showed that the exchange factor method gives the same results as the zonal method for grey media \citep{Larsen1983, LarsenHowell1985}, the equivalence extends to the exchange factor method as well, whose exchange factors involve the inverse of the gas-to-gas block and are therefore in general dense like Noble's total exchange areas. The proposed formulation therefore does not change the solution of Hottel's method, but differs from its usual form in two respects.

First, the partitioning $\mathbf{A} = \mathbf{F}\circ(\mathbf{1}-\mathbf{B})$, $\mathbf{R} = \mathbf{F}\circ\mathbf{B}$ with the column-constant matrix $\mathbf{B}$ leads to a single linear solve whose system matrix inherits the sparsity of $\mathbf{F}$. Noble's total exchange areas, by contrast, are dense whenever multiple reflection-scattering connects all elements (Proposition~\ref{prop:total_exchange_positivity}), so that their storage grows as $(m+n)^2$ even when $\mathbf{F}$ is sparse. Hottel's balances solved directly for the radiosities share the sparsity of the proposed formulation (Section~\ref{sec:total_exchange_areas}), but the zonal method is usually applied through the total exchange areas; the proposed formulation takes the sparse route directly, with Theorem~\ref{thm:M_nonsingular} ensuring that the single solve is well posed for any admissible combination of boundary conditions. Section~\ref{sec:large_scale} shows what this enables: a problem with about $10^6$ elements, for which the gas-gas block of the total exchange areas alone would require 8.0~TB of dense storage once the medium scatters, is solved in about one minute.

Second, the guarantees do not rely on reciprocity. The exact exchange factors satisfy two constraints: closure, equations \eqref{eq:viewfactor_surfsurf}--\eqref{eq:viewfactor_gasgas} normalized so that every row of $\mathbf{F}$ sums to unity, and reciprocity, equation \eqref{eq:reciprocity}. The symmetric notation of the zonal method, and the closed-form total exchange areas built on it, assume both. The two constraints behave differently under estimation, however. Closure is a property of every set of traced rays, since each ray has its first interaction somewhere, so the normalization of each row by its number of rays makes it exact for any finite sample. Reciprocity is a symmetry of the exact kernel, which every finite estimate violates by its statistical noise. The power-variable formulation makes visible that the two play different roles: non-singularity (Theorem~\ref{thm:M_nonsingular}), non-negativity (Theorem~\ref{thm:physical_correctness}) and unconditional energy conservation (Theorem~\ref{thm:unconditional_conservation}) rest on closure alone, and hold for any non-negative, row-stochastic and irreducible $\mathbf{F}$, including ray-traced exchange factors as traced; by Theorem~\ref{thm:hottel_equivalence} they hold for Hottel's balances with the same exchange factors (Corollary~\ref{cor:hottel_properties}). Reciprocity is needed only where the direct exchange areas are used as a symmetric matrix (Remark~\ref{rem:nonreciprocal}). It thereby becomes a matter of accuracy, which smoothing improves, rather than a precondition for a physically admissible solution, and the exchange factors can be used without smoothing, as is done throughout Section~\ref{sec:numerical_validation}.

For transparent enclosures with gray-diffuse surfaces and no participating medium, the proposed formulation specializes to classical radiosity in power form: with $\mathbf{F}_{sg} = \mathbf{F}_{gs} = \mathbf{F}_{gg} = \mathbf{0}$ and column-constant $\rho$, the equation $\mathbf{D}\mathbf{j} = \mathbf{e}$ becomes $(\mathbf{I} - \rho\mathbf{F}^\top)\mathbf{j} = \mathbf{e}$, which is the classical radiosity equation. Section~\ref{sec:symbolic_validation} verifies this specialization symbolically for the canonical cases of infinite parallel plates and concentric cylinders, recovering the textbook closed-form formulas exactly as algebraic identities.

Methods that solve the RTE directly treat the angular dependence of the intensity in one of two ways. Spherical harmonics and other moment methods expand it in a series, which is truncated and closed. The discrete ordinates and finite volume methods instead partition the directions into a finite set, much as the zonal method partitions the domain into zones, and converge as this angular discretization is refined. The proposed formulation, like the zonal method, discretizes only the spatial domain: for diffuse surfaces and isotropic scattering, the angular dependence is integrated within the exchange factors, analytically or by ray tracing, and the multi-bounce paths are captured through the matrix inverse $(\mathbf{I} - \mathbf{R}^\top)^{-1}$ implicit in the linear solve. Its zoning thus plays the role that the angular discretization plays in the discrete ordinates method.

The proposed formulation may be viewed as an augmented form of Monte Carlo. For highly reflecting-scattering domains, such as clouds or snow, the analytical path tracing avoids tracing rays through repeated scattering events. The timings in Section~\ref{sec:large_scale} illustrate the cost of the two steps: for about $10^6$ elements, tracing about $10^9$ first-interaction rays took under 30 seconds, and the single linear solve about one minute. Conventional Monte Carlo for a strongly scattering medium would require tracing each ray through numerous scattering events: as $\omega \to 1$, a ray is removed only by absorption at the walls, and the number of scattering events per ray grows with the optical thickness, roughly as $\tau_L^2$ in the diffusive regime, multiplying the ray tracing cost accordingly. The proposed method instead traces rays only to first interaction, then resolves all subsequent scattering through the matrix framework. Additionally, once $\mathbf{F}$ is computed, problems with different combinations of $\kappa$ and $\sigma_s$ summing to the same $\beta$ can be solved without further ray tracing, enabling rapid parameter studies.

Furthermore, the simplified first-interaction ray tracing provides a more accessible pathway to the programming of radiative transfer solvers, while simultaneously making the ray tracing highly suitable for both CPU and GPU parallel execution. Unlike conventional multiple scattering ray tracing, where rays terminate after unpredictable numbers of bounces creating workload imbalance, first-interaction tracing has predictable and uniform workload per ray. This also provides a validation pathway for developers of new Monte Carlo solvers: by implementing first-interaction tracing and applying the proposed formulation, one obtains multiple reflection-scattering solutions that are exact for the traced exchange factors, and that can serve as reference solutions for full Monte Carlo implementations on arbitrary geometries, with the statistical error of $\mathbf{F}$ and the zoning as the only expected differences.

\subsection{Limitations and Future Work}
\label{sec:limitations}

Limitations of the proposed formulation include the general need for ray tracing. While Monte Carlo ray tracing can achieve arbitrary accuracy with quantifiable bounds, its convergence rate is proportional to $N_\mathrm{b}^{-1/2}$, meaning high ray sample counts are preferred. To preserve accuracy under increased discretization the number of ray samples per cell should remain constant. For high resolution domains with several elements, this leads to an unfortunate scaling of the total necessary ray bundles. Therefore, to harness the full potential of the proposed formulation when applied to larger problems, some form of variance reduction should be applied. The two established smoothing approaches each have a limitation in this setting. The least-squares smoothing of Larsen and Howell \citep{LarsenHowell1986} requires only a direct linear solve, but does not enforce non-negativity, so that small exchange factors may become negative, and truncating them afterwards reintroduces reciprocity and closure errors. The constrained maximum likelihood estimation of Daun et al. \citep{Daun2005} enforces reciprocity, closure and non-negativity together, but requires a nonlinear optimizer over the exchange factors, whose cost grows rapidly with the number of elements. A composite smoother based on projections onto the reciprocity and closure constraints, which preserves non-negativity, is implemented in the package used for the numerical validation \citep{Bielefeld2024}, where it also provides the reciprocal exchange factors for the comparison with Noble's formulation in Section~\ref{sec:hottel_numerical}. Its analysis is left to future work.

For large-scale problems with global radiative connectivity, the proposed formulation is limited primarily by the memory requirements of the dense matrices $\mathbf{A}$, $\mathbf{R}$, $\mathbf{C}$, and $\mathbf{D}$ when $\mathbf{F}$ is dense. To extend applicability to such problems, domain decomposition followed by sequential iterative solution where the boundary fluxes match between decomposed parts is a natural direction. The proposed formulation lends itself to this type of decomposition, since discrete angular in-flux information can be obtained from the entries of the entry-wise product $(\mathbf{A}^\top+\mathbf{R}^\top)_i \circ \mathbf{j}$, where the index $i$ denotes the row, meaning this vector gives the in-flux of power onto row element $i$ from each column element $j$ in the domain. A general algorithm for accomplishing such a decomposed solution would be a valuable extension.

If one has a priori knowledge of the problem to be solved, one may sample the domain selectively for $\mathbf{F}$, for example based on the emissive power for non-reflecting non-scattering problems, or based on the previous $\mathbf{j}$ for general dynamic problems, rather than the uniform sampling which was employed in this paper. This could potentially allow a significant reduction in the number of samples necessary to accurately resolve the radiative transfer field. This importance sampling would be especially relevant for dynamic problems, where $\mathbf{F}$ needs to be updated for each time step, for example for convective transport of a fluid with spatially and temporally varying extinction $\beta$.

As presented here, the formulation is grey and assumes diffuse emission and reflection at surfaces and isotropic scattering in the medium, since the angular distributions are fixed by those used to trace $\mathbf{F}$ (uniform for volumes and Lambertian for surfaces). Both restrictions are beyond the scope of this paper, but extensions are implemented in the package used for the numerical validation \citep{Bielefeld2024}: a spectral extension, in which the radiative properties are resolved into spectral bins coupled through the element temperatures, and, for two-dimensional participating media, a directional extension for anisotropic scattering and non-diffuse reflection, in which the exchange factors are resolved into angular bins by the direction of travel. Their analysis and validation are left to future work.

\section{Conclusion}
\label{sec:conclusion}

This paper presented a matrix formulation of the scalar laws of radiative energy balance for coupled mixed boundary condition radiative transfer problems on general domains. Starting from a non-dimensional first-interaction exchange factor matrix $\mathbf{F}$, the formulation partitions $\mathbf{F}$ via Hadamard products into single-step absorption and reflection-scattering matrices, from which a mixed-boundary linear system is assembled row by row. A discrete minimum principle shows that this system is non-singular whenever $\mathbf{F}$ is irreducible and at least one element with prescribed temperature absorbs radiation, and that its solution is non-negative for non-negative prescribed sources. Energy conservation holds unconditionally as an algebraic identity, the condition number is bounded independently of the mesh when all temperatures are prescribed, and radiative equilibrium solutions are invariant under the single scattering albedo. None of these results requires reciprocal exchange factors, and the minimum principle with its two consequences, non-singularity and non-negativity, has been machine-checked in Lean 4.

The formulation was proven to be Hottel's zonal balances in power variables, for any exchange factor matrix, so that its guarantees are guarantees for Hottel's method, and for reciprocal exchange factors Noble's total exchange areas were recovered from it in closed form. The two were confirmed to agree symbolically and, on a larger system, numerically to machine precision. They differ in cost: the total exchange areas are dense whenever multiple reflection-scattering connects all elements, whereas the balances in power variables are a single linear solve whose system matrix inherits the sparsity of $\mathbf{F}$.

Validation was provided symbolically against the textbook closed-form solutions for infinite parallel plates and concentric cylinders, recovering both as exact algebraic identities, and numerically against the results of Crosbie and Schrenker for pure and partial scattering, against an exact symmetry identity in a three-dimensional enclosure, and against the diffusion approximation in an optically thick problem with about $10^6$ elements, which was solved in about one minute. With ray-traced exchange factors, the accuracy of the solution decreases as $N_\mathrm{b}^{-1/2}$ in the number of rays, with Sobol sampling lowering the error by a constant factor. At large scale the accuracy is therefore limited by the exchange factors rather than by the formulation, which makes variance reduction for exchange factors the natural next step.

\section*{Conflict Of Interest}

The author acknowledges previous employment as an intern at the boiler and power plant company Aalborg Energie Technik A/S during autumn 2022, while completing the 9th semester of his Master's in Thermal Energy and Process Engineering at Aalborg University, and current employment at the engineering consulting firm Added Values P/S. This research was conducted independently, outside of regular work hours, and without financial or material support from either organization. The views and conclusions presented in this paper represent the author's personal work and do not necessarily reflect the positions or policies of the aforementioned employers.

\section*{Funding}

This research did not receive any specific grant from funding agencies in the public, commercial, or not-for-profit sectors.

\section*{Declaration of generative AI and AI-assisted technologies in the manuscript preparation process}

This work was developed over more than two years, from mid-2024 onwards, as one continuous collaboration between the author and a succession of Anthropic's Claude models, as described in the following.

During the preparation of this work the author used Anthropic's Claude 3.5 Sonnet and later Claude Sonnet 4.0. These models assisted at all stages of this research, from method development through manuscript preparation. In particular, Claude 3.5 Sonnet facilitated iterative development and refinement of key concepts and arguments, which led to the central idea of the method in the summer of 2024. Claude Sonnet 4.0 assisted in formulating the initial proofs of physical correctness. Claude Opus 4.5 assisted with final manuscript review and language refinement of the initial arXiv version. The author used Anthropic's Claude Opus 4.7 for assistance in identifying the methodology-error in the first version of this article, through symbolic validation against canonical closed-form solutions.

For the third version, the work was carried out in extended dialogue with Anthropic's Claude models (Claude Opus 5.5, Claude Fable 5 and Claude Fable 5.1), in which the author posed the questions and directed the checks, and the models proposed analyses, derivations and code that the author examined, tested, and accepted or rejected. The contributions of the models to the content of this version are the following. Prompted by the author's questions, Claude Opus 5.5 identified the input-conversion error in the author's comparison with Noble's formulation of Hottel's zonal method, which had produced the apparent discrepancy reported in the first and second versions, and derived, under the author's direction, the equivalence with Hottel's zonal balances (Theorem~\ref{thm:hottel_equivalence}), the recovery of Noble's total exchange areas (Proposition~\ref{prop:total_exchange_areas}), the conditioning bound (Corollary~\ref{cor:D_conditioning}) and the albedo invariance (Proposition~\ref{prop:albedo_invariance}). The gap in the second version's proof of non-singularity was found by the Lean 4 formalization of that version, drafted by Claude Fable 5; the idea of replacing the proof by a discrete maximum principle on the masked system originated with the model, as did the symmetry argument for the identity of Section~\ref{sec:3d}. The decision to derive both non-singularity and non-negativity from a single minimum principle (Lemma~\ref{lem:min_principle}), and the ordering of the argument in Section~\ref{sec:min_principle}, are the author's; the proofs as printed were then drafted by Claude Fable 5.1 and checked by the author. The Lean 4 modules of this version were drafted by Claude Fable 5.1 in collaboration with the author, who compiled and machine-checked every file against Mathlib. The scripts that reproduce the figures of this version were written with the same assistance, and most of the manuscript text of this version was drafted by Claude under the author's direction and subsequently revised and verified by the author. The method itself, the choice of what to prove and what to validate, the decisions recorded in \ref{app:errata}, and the final judgement on every result remain the author's.

Every result in this version was verified independently of the models, symbolically, numerically or by machine-checked proof, and the author remains solely accountable for the content of this article.

\section*{Acknowledgements}

The author wishes to acknowledge the following assistance: His friends Snorre Marstrand Enevoldsen and Sune Borkfelt who helped by proof reading the English language of the draft manuscript. The author's former supervisors at Aalborg University, Thomas Joseph Condra and Kim Sørensen, as well as the author's current colleague Gorm Bruun Andresen who all helped by proof reading the English language of the draft manuscripts and by suggesting improvements to the presentation. The financial support from the author's parents during his education. The author is grateful to Professor Ryan G. McClarren of the Department of Aerospace and Mechanical Engineering, University of Notre Dame, for endorsing the first version of this manuscript for submission to arXiv, which allowed the work to continue at a time when it otherwise might not have.

An earlier version of this manuscript was submitted to a peer-reviewed journal. The author is grateful to the anonymous reviewers whose thoughtful feedback and constructive suggestions substantially strengthened this work, leading to a more comprehensive validation framework and enhanced comparison to existing methods.

\section*{Software}

All calculations were performed using the Julia programming language \citep{Bezanson2017}, with the Julia packages RayTraceHeatTransfer.jl \citep{Bielefeld2024} for the exchange factors and linear solves, Sobol.jl \citep{JohnsonSobol} for the Sobol sequences used in the ray tracing, Krylov.jl \citep{Montoison2023} for the iterative solution of the large-scale system, and ConvolutionInterpolations.jl \citep{Bielefeld2025} for the interpolation of tabulated reference solutions. All figures were created using the Julia package Makie.jl \citep{DanischKrumbiegel2021}, and the scripts producing them are available in the repository of RayTraceHeatTransfer.jl.

\appendix

\section{Symbolic Validation Scripts}
\label{app:symbolic_scripts}

This appendix provides the SymPy \citep{SymPy2017} scripts used for the symbolic validation against infinite parallel plates and concentric cylinders presented in Section~\ref{sec:symbolic_validation}. The scripts are self-contained and require only a standard SymPy installation. They produce all symbolic identities and matrix expressions reported in Sections \ref{sec:symbolic_plates} and \ref{sec:symbolic_cylinders}.

\subsection{Common Helper Function}
\label{app:helpers}

The following helper function constructs the system matrices under the single-step definitions. It is reused in both validation cases.

\begin{verbatim}
import sympy as sp
from sympy import Matrix, symbols, simplify, factor, eye

def build_matrices(F_view, B):
    """Build A, R, C, D under the single-step formulation."""
    n = F_view.shape[0]
    A = Matrix([[(1 - B[i, j]) * F_view[i, j] for j in range(n)]
                for i in range(n)])
    R = Matrix([[B[i, j] * F_view[i, j] for j in range(n)]
                for i in range(n)])
    A = simplify(A)
    R = simplify(R)
    C = simplify(eye(n) - A.T - R.T)
    D = simplify(eye(n) - R.T)
    return A, R, C, D
\end{verbatim}

\subsection{Parallel Plates}
\label{app:plates_script}

\begin{verbatim}
rho, sigma_, T1, T2, A = symbols('rho sigma T_1 T_2 A', positive=True)
eps = 1 - rho
Eb1 = sigma_ * T1**4
Eb2 = sigma_ * T2**4
e1 = eps * Eb1 * A
e2 = eps * Eb2 * A

F_view = Matrix([[0, 1], [1, 0]])
B = Matrix([[rho, rho], [rho, rho]])

A_mat, R_mat, C_mat, D_mat = build_matrices(F_view, B)

# Solve mixed system M j = h with M = D (both temperatures known)
j = simplify(D_mat.solve(Matrix([e1, e2])))
g_a = simplify(A_mat.T * j)
q1 = simplify(e1 - g_a[0])
q2 = simplify(e2 - g_a[1])

# Textbook reference formula
Q_textbook = (Eb1 - Eb2) * A / (1/eps + 1/eps - 1)
assert simplify(q1 - Q_textbook) == 0  # exact symbolic match
assert simplify(q1 + q2) == 0          # energy conservation
\end{verbatim}

\subsection{Concentric Cylinders}
\label{app:cylinders_script}

\begin{verbatim}
rho, sigma_, T1, T2, A1, alpha = symbols(
    'rho sigma T_1 T_2 A_1 alpha', positive=True)
eps = 1 - rho
A2 = A1 / alpha
Eb1 = sigma_ * T1**4
Eb2 = sigma_ * T2**4
e1 = eps * Eb1 * A1
e2 = eps * Eb2 * A2

F_view = Matrix([[0,     1      ],
                 [alpha, 1-alpha]])
B = Matrix([[rho, rho], [rho, rho]])

A_mat, R_mat, C_mat, D_mat = build_matrices(F_view, B)

j = simplify(D_mat.solve(Matrix([e1, e2])))
g_a = simplify(A_mat.T * j)
q1 = simplify(e1 - g_a[0])
q2 = simplify(e2 - g_a[1])

Q_textbook = A1 * sigma_ * (T1**4 - T2**4) / (
    1/eps + alpha * (1/eps - 1))
assert simplify(q1 - Q_textbook) == 0
assert simplify(q1 + q2) == 0
\end{verbatim}

The above assertions all evaluate to True when executed, confirming the symbolic agreement reported in Section~\ref{sec:symbolic_validation}.

\section{Errata and Changes from Previous Versions}
\label{app:errata}

This appendix records the errors corrected in each version of the manuscript, so that readers of earlier versions can identify which results have changed.

\subsection{Changes from Version 1}
\label{app:errata_v1}

The second version of the manuscript (arXiv:2512.22157v2) corrected a structural error in the definitions of the absorption and reflection-scattering matrices in v1 (arXiv:2512.22157v1, December 2025). In v1, the matrices were defined via the steady-state path matrix $\mathbf{S}_\infty = (\mathbf{I} - \mathbf{K})^{-1}\mathbf{F}$ as
\begin{equation*}
\mathbf{A}_{\text{v1}} = (\mathbf{1}-\mathbf{B})^\top \circ \mathbf{S}_\infty \circ (\mathbf{1}-\mathbf{B}),
\quad
\mathbf{R}_{\text{v1}} = (\mathbf{1}-\mathbf{B})^\top \circ \mathbf{S}_\infty \circ \mathbf{B},
\end{equation*}
embedding the full multi-bounce Neumann sum inside the operator. When the resulting $\mathbf{R}_{\text{v1}}^\top$ was then used inside $\mathbf{D}_{\text{v1}} = \mathbf{I} - \mathbf{R}_{\text{v1}}^\top$ to solve the linear system $\mathbf{M}\mathbf{j} = \mathbf{h}$, the Neumann sum was applied twice: once in $\mathbf{R}_{\text{v1}}$ itself, and again through the matrix inverse $(\mathbf{I} - \mathbf{R}_{\text{v1}}^\top)^{-1}$. For all-temperatures-prescribed configurations with reflecting walls, this yielded results that differed quantitatively from classical radiosity, while still satisfying all the algebraic identities of the v1 theorems. The v1 numerical validation cases (diffusion approximation, Crosbie-Schrenker) all used non-reflecting walls ($\rho=0$), so $\mathbf{R}_{\text{v1}} = \mathbf{0}$ and the double application of the Neumann sum was identically zero. The error was therefore invisible to the v1 validation suite.

The error was identified through symbolic validation against the textbook closed-form solutions for infinite parallel plates and concentric cylinders (Section~\ref{sec:symbolic_validation}). The v1 definitions failed to reproduce these formulas for $\rho > 0$, while the v2 definitions recover them as exact algebraic identities. The symbolic appendix that found the error has therefore been retained as a new validation tier in v2.

The corrected definitions used in this version are
\begin{equation*}
\mathbf{A} = \mathbf{F} \circ (\mathbf{1}-\mathbf{B}),
\quad
\mathbf{R} = \mathbf{F} \circ \mathbf{B},
\end{equation*}
which represent single-step absorption and reflection-scattering interactions. The multi-bounce structure is then generated implicitly by the matrix inverse $(\mathbf{I} - \mathbf{R}^\top)^{-1}$ when the linear system is solved, which is structurally analogous to classical radiosity and consistent with the canonical reference results validated in Section~\ref{sec:symbolic_validation}. The title of v1 used the word \textit{transformation} to describe the v1 augmentation of $\mathbf{F}$ via $\mathbf{S}_\infty$; v2 uses \textit{formulation} to describe the simpler single-step partition more honestly.

The corrected formulation preserves the central theorems of v1 with simpler proofs, and brings the following structural improvements:
\begin{enumerate}
\item The identity $\mathbf{A} + \mathbf{R} = \mathbf{F}$ is now immediate from the single-step definitions and the row-stochasticity of $\mathbf{F}$, replacing the multi-step derivation through $\mathbf{F}(\mathbf{1}-\mathbf{b}) = (\mathbf{I}-\mathbf{K})\mathbf{1}$ used in v1. The Perron eigenvector of $\mathbf{A}+\mathbf{R}$ becomes $\mathbf{1}$ rather than $\mathbf{1}-\mathbf{b}$.
\item Energy conservation (Theorem~\ref{thm:unconditional_conservation}) is now unconditional. The v1 result required uniform reflection-scattering coefficients $\gamma$, with an additional theorem covering the non-uniform case under specific row-selection hypotheses. The corrected formulation gives $\mathbf{1}^\top\mathbf{C} = \mathbf{0}^\top$ as an immediate consequence of $\mathbf{A} + \mathbf{R} = \mathbf{F}$ being row-stochastic, holding for arbitrary $\mathbf{B}$ and arbitrary boundary configurations.
\item The matrix $\mathbf{C} = \mathbf{I} - \mathbf{F}^\top$ is now purely geometric, depending only on the exchange factor matrix and not on the optical properties of the surfaces or medium. Optical properties enter only through $\mathbf{D}$. For parameter studies over optical properties at fixed geometry, $\mathbf{C}$ need only be assembled once.
\item The solution of the radiative transfer problem now requires only one linear solve, $\mathbf{M}\mathbf{j} = \mathbf{h}$, rather than two. The v1 method required first computing $\mathbf{S}_\infty = (\mathbf{I} - \mathbf{K})^{-1}\mathbf{F}$ to construct the matrices, then solving for $\mathbf{j}$. The single solve preserves any sparsity present in $\mathbf{F}$, which the v1 construction of $\mathbf{S}_\infty$ destroyed.
\item The conditioning analysis of the linear system in v2 distinguished the all-temperatures-prescribed configuration $\mathbf{M}=\mathbf{D}$ from the mixed-boundary configuration with radiative-equilibrium gas rows from $\mathbf{C}$. That analysis measured conditioning by determinants and has been withdrawn in the third version, see \ref{app:errata_v2}.
\end{enumerate}

The convergence theorems of v1 Section 3 have been removed. The Neumann-series convergence theorem characterized $\mathbf{S}_\infty$, which no longer appears in the formulation. The numerical stability theorem stated a Gershgorin-derived determinant bound, $|\det(\mathbf{I}-\mathbf{K})| \geq (1-\gamma)^{m+n}$ for uniform $\gamma$, and recast it as an a-priori precision requirement on $m+n$ given machine precision $\epsilon$. The bound itself is a valid consequence of Gershgorin's disk theorem, but the recasting as a necessary precision criterion conflates a worst-case eigenvalue configuration with the typical behaviour of $\mathbf{I}-\mathbf{K}$ for ray-traced $\mathbf{F}$. The eigenvalues of $\mathbf{I}-\mathbf{K}$ rarely sit simultaneously at the leftmost Gershgorin positions in actual geometries, so the bound is far from tight and presenting it as a discretization limit is misleading. In v2, conditioning was instead assessed empirically by determinants measured from ray-traced $\mathbf{F}$, an approach that is itself corrected in the third version, see \ref{app:errata_v2}.

The numerical validation cases of v1 (diffusion approximation, Crosbie-Schrenker) used non-reflecting walls, where the v1 and v2 formulations agree. Those validations therefore remain valid for the corrected formulation. The discrepancy in Noble's formulation of Hottel's zonal method reported in v1 and v2 was not affected by this correction, but was itself an error in the comparison, see \ref{app:errata_v2}.

\subsection{Changes from Version 2}
\label{app:errata_v2}

The third version corrects two errors in v2 (arXiv:2512.22157v2) and repairs two of its proofs.

\begin{enumerate}
\item \textit{Comparison with Noble's formulation.} Versions 1 and 2 reported that Noble's matrix formulation of Hottel's zonal method produces non-zero element-wise source terms in a medium in radiative equilibrium at intermediate single scattering albedo. This discrepancy was an error in the author's comparison, not in Noble's formulation: the emitted power of the gas obtained with the proposed formulation was converted to a blackbody emissive power by dividing by $4\beta V$ instead of $4\kappa V$, which understates the blackbody emissive power by the factor $1-\omega$. With the correct conversion (Remark~\ref{rem:temperature_recovery}), the two formulations agree exactly, symbolically for all albedos (Section~\ref{sec:hottel_symbolic}) and numerically to machine precision (Section~\ref{sec:hottel_numerical}). The third version proves this agreement in general: the proposed formulation and Hottel's zonal balances are the same discrete model for any exchange factor matrix (Section~\ref{sec:hottel_equivalence}), and for reciprocal exchange factors Noble's total exchange areas follow from it in closed form. The difference between the two is one of variables, cost and sparsity, not of the solution.

\item \textit{Conditioning.} Version 2 assessed the conditioning of the linear system by its determinant, measured from ray-traced exchange factors, and described a conditioning frontier that limits the admissible number of elements, optical thickness and scattering. The determinant is not a measure of conditioning: it is the product of the eigenvalues and decreases with the number of elements even when the system is well conditioned. The third version replaces this analysis with a proven bound, $\operatorname{cond}_1(\mathbf{D}) \leq (1+b_{\max})/(1-b_{\max})$, independent of the number of elements, the geometry and the optical thickness (Corollary~\ref{cor:D_conditioning}), and with the measured condition number of the mixed-boundary system, which remains moderate for meshes and optical thicknesses where v2 suggested a limit (Section~\ref{sec:conditioning}).

\item \textit{Proofs of non-singularity and non-negativity.} The machine-checked formalization in Lean 4 revealed a gap in the v2 proof of the non-singularity of the mixed-boundary system. The third version replaces that proof with a discrete minimum principle (Lemma~\ref{lem:min_principle}, Theorem~\ref{thm:M_nonsingular}), which requires no spectral input and only one element with prescribed temperature that absorbs radiation. The v2 proof of non-negativity concluded that the mixed-boundary matrix is an M-matrix from its sign pattern and its non-singularity, which does not follow in general; the third version derives non-negativity from the same minimum principle (Theorem~\ref{thm:physical_correctness}), without spectral arguments or reciprocity. The statements of both results are essentially unchanged, and both have been machine-checked, together with the minimum principle from which they follow.
\end{enumerate}

Further changes in the third version, which do not correct errors, are the following. The matrix of emission capacities is denoted $\mathbf{W}$ in place of $\mathbf{E}$, and the title now states the results and the method they concern, since the third version identifies the formulation with Hottel's zonal balances. The theorems of v2 on the M-matrix properties of $\mathbf{C}$ and $\mathbf{D}$ have been removed, since the minimum principle establishes non-singularity and non-negativity without them. The Introduction has been expanded with a review of the zonal, exchange factor and Monte Carlo literature. New results are the conditioning bound, the albedo invariance of radiative equilibrium solutions (Proposition~\ref{prop:albedo_invariance}) and the equivalence with Hottel's zonal balances. The numerical validation has been rebuilt so that every figure is produced by a script in the repository of the package used, with exchange factors used as traced, without smoothing, except for the comparison with Noble's formulation. The uncertainty propagation of v2 is replaced by the spread of independent replicates, and Sobol sampling is compared with pseudorandom sampling. The three-dimensional cube example is replaced by an exact symmetry identity on an icosphere, the star-shaped example has been removed, and the diffusion and medium-scale examples are replaced by a single optically thick problem with about $10^6$ elements. Mesh convergence studies have been deferred to future work on exchange factor smoothing.


\begin{thebibliography}{00}

\bibitem{Kirchhoff1860}
	Kirchhoff, G.,
	\textit{Ueber das Verhältniss zwischen dem Emissionsvermögen und dem Absorptionsvermögen der Körper für Wärme und Licht},
	Annalen der Physik und Chemie,
	109, 275-301,
	1860.
	
\bibitem{Schwarzschild1906}
	Schwarzschild, K.,
	\textit{Ueber das Gleichgewicht der Sonnenatmosphäre},
	Nachrichten von der Gesellschaft der Wissenschaften zu Göttingen, Mathematisch-Physikalische Klasse,
	1906, 41-53,
	1906.
	
\bibitem{Chandrasekhar1950}
	Chandrasekhar, S.,
	\textit{Radiative Transfer},
	Clarendon Press, Oxford,
	International Series of Monographs on Physics,
	1950.

\bibitem{Daun2021}
	Howell, J. R. and Mengü\c{c}, M. P. and Daun, K. J. and Siegel, R.,
	\textit{Thermal Radiation Heat Transfer},
	CRC Press, Taylor \& Francis Group,
	7th edition,
	2021.
  
\bibitem{Modest2022}
	Modest, M. F. and Mazumder, S.,
	\textit{Radiative Heat Transfer},
	Academic Press,
	4th edition,
	2022.
	
\bibitem{Hottel1958}
	Hottel, H. C. and Cohen, E. S.,
	\textit{Radiant Heat Exchange In A Gas-Filled Enclosure: Allowance For Nonuniformity Of Gas Temperature},
	A.I.Ch.E., Vol. 4, No. 1,
	March, 1958,
	\url{https://doi.org/10.1002/aic.690040103}.
		
\bibitem{Hottel1967}
	Hottel, H. C. and Sarofim, A. F.,
	\textit{Radiative Transfer},
	McGraw-Hill Book Company,
	1st edition,
	1967.
	
\bibitem{Noble1975}
	Noble, J. J.,
	\textit{The Zone Method: Explicit Matrix Relations For Total Exchange Areas}
	Int. J. Heat Mass Transfer.,
	Vol 18, pp. 261-269, 1975,
	\url{https://doi.org/10.1016/0017-9310(75)90158-1}.

\bibitem{Naraghi1985}
	Naraghi, M. H. N. and Chung, B. T. F.,
	\textit{A Unified Matrix Formulation for the Zone Method: A Stochastic Approach},
	Int. J. Heat Mass Transfer,
	Vol. 28, No. 1, pp. 245-251, 1985,
	\url{https://doi.org/10.1016/0017-9310(85)90026-2}.

\bibitem{Goyheneche2002}
	Goyh\'en\`eche, J. M. and Sacadura, J. F.,
	\textit{The Zone Method: A New Explicit Matrix Relation to Calculate the Total Exchange Areas in Anisotropically Scattering Medium Bounded by Anisotropically Reflecting Walls},
	ASME Journal of Heat Transfer,
	Vol. 124, No. 4, pp. 696-703, 2002,
	\url{https://doi.org/10.1115/1.1481359}.

\bibitem{Gebhart1961}
	Gebhart, B.,
	\textit{Surface Temperature Calculations in Radiant Surroundings of Arbitrary Complexity -- for Gray, Diffuse Radiation},
	Int. J. Heat Mass Transfer,
	1961,
	\url{https://doi.org/10.1016/0017-9310(61)90048-5}.

\bibitem{Larsen1983}
	Larsen, M. E.,
	\textit{The Exchange Factor Method: An Alternative Zonal Formulation for Analysis of Radiating Enclosures Containing Participating Media},
	PhD dissertation, The University of Texas at Austin,
	1983,
	\url{https://www.osti.gov/biblio/5004297}

\bibitem{LarsenHowell1985}
	Larsen, M. E. and Howell, J. R.,
	\textit{The Exchange Factor Method: An Alternative Basis for Zonal Analysis of Radiating Enclosures},
	ASME Journal of Heat Transfer,
	Vol. 107, No. 4, pp. 936-942, 1985,
	\url{https://doi.org/10.1115/1.3247524}.

\bibitem{LarsenHowell1986}
	Larsen, M. E. and Howell, J. R.,
	\textit{Least-Squares Smoothing of Direct-Exchange Areas in Zonal Analysis},
	ASME Journal of Heat Transfer,
	Vol. 108, No. 1, pp. 239-242, 1986,
	\url{https://doi.org/10.1115/1.3246898}.
	
\bibitem{MetropolisUlam1949}
	Metropolis, N. and Ulam, S.,
	\textit{The Monte Carlo Method},
	J. Am. Stat. Assoc.,
	44(247), pp. 335-341, September,
	1949,
	\url{https://doi.org/10.1080/01621459.1949.10483310}.

\bibitem{HowellPerlmutter1964}
	Howell, J. R. and Perlmutter, M.,
	\textit{Monte Carlo Solution of Thermal Transfer Through Radiant Media Between Gray Walls},
	ASME Journal of Heat Transfer,
	Vol. 86, No. 1, pp. 116-122, 1964,
	\url{https://doi.org/10.1115/1.3687044}.
	
\bibitem{Narayanaswamy2015}
	Narayanaswamy, A.,
	\textit{An Analytic Expression For Radiation View Factor Between Two Arbitrarily Oriented Planar Polygons},
	International Journal of Heat and Mass Transfer 91 (2015) 841–847,
	\url{https://doi.org/10.1016/j.ijheatmasstransfer.2015.07.131}.

\bibitem{AbramowitzStegun1964}
	Abramowitz, M. and Stegun, I. A. (Eds.),
	\textit{Handbook of Mathematical Functions with Formulas, Graphs, and Mathematical Tables},
	National Bureau of Standards, Applied Mathematics Series 55,
	Section 27.8, Clausen's Integral and Related Summations, p. 1005,
	1964.

\bibitem{BrentHarvey2013}
	Brent, R. P. and Harvey, D.,
	\textit{Fast Computation of Bernoulli, Tangent and Secant Numbers},
	in: Bailey, D. H. et al. (Eds.), Computational and Analytical Mathematics,
	Springer Proceedings in Mathematics \& Statistics, Vol. 50, Springer, New York,
	pp. 127-142, 2013,
	\url{https://doi.org/10.1007/978-1-4614-7621-4_8}.

\bibitem{Sobol1967}
	Sobol', I. M.,
	\textit{On the Distribution of Points in a Cube and the Approximate Evaluation of Integrals},
	USSR Computational Mathematics and Mathematical Physics 7 (4) (1967) 86–112,
	\url{https://doi.org/10.1016/0041-5553(67)90144-9}.

\bibitem{JoeKuo2008}
	Joe, S. and Kuo, F. Y.,
	\textit{Constructing Sobol Sequences with Better Two-Dimensional Projections},
	SIAM Journal on Scientific Computing 30 (5) (2008) 2635–2654,
	\url{https://doi.org/10.1137/070709359}.

\bibitem{JohnsonSobol}
	Johnson, S. G.,
	\textit{Sobol.jl: Generation of Sobol Low-Discrepancy Sequences for the Julia Language},
	\url{https://github.com/JuliaMath/Sobol.jl}.
	
\bibitem{Daun2005}
	Daun, K. J. and Morton, D. P. and Howell, J. R.,
	\textit{Smoothing Monte Carlo Exchange
Factors Through Constrained
Maximum Likelihood Estimation},
	ASME Journal of Heat Transfer,
	October, 2005,
	\url{https://doi.org/10.1115/1.2035111}.

\bibitem{BornWolfOptics1980}
	Born, M. and Wolf, E.,
	\textit{Principles of Optics - Electromagnetic Theory of Propagation, Interference and Diffraction of Light}
	Pergamon Press,
	Sixth (corrected) edition,
	1980.

\bibitem{MonteCarloPastFuture2021}
	Howell, J. R and Daun, K. J.,
	\textit{The Past and Future of the Monte Carlo Method in Thermal Radiation Transfer}
	ASME Journal of Heat Transfer,
	Vol. 143,
	October 2021,
	\url{https://doi.org/10.1115/1.4050719}.
	
\bibitem{Bezanson2017}
	Bezanson, J. and Edelman A. and Karpinski, S., and
Shah, V. B.
	\textit{Julia: A Fresh Approach to Numerical Computing}
	SIAM Review.
	Vol. 59, Iss. 1 (2017),
	\url{https://doi.org/10.1137/141000671}.
	
\bibitem{HornJohnson2018}
	Horn, R. A. and Johnson, C. R.,
	\textit{Matrix Analysis},
	Cambridge University Press,
	2nd edition,
	2018.
	
\bibitem{Crosbie1984}
	Crosbie, A. L. and Schrenker, R. G.,
	\textit{Radiative Transfer In A Two-Dimensional Rectangular Medium Exposed To Diffuse Radiation}
	J. Quant. Spectrosc. Radiat. Transfer,
	Vol. 31, No. 4, pp. 339-372, 1984,
	\url{https://doi.org/10.1016/0022-4073(84)90095-5}.

\bibitem{Wagner2021}
	Kerkhoff, J. A. and Wagner, M. J.,
	\textit{viewFactor.m},
	\url{https://github.com/uw-esolab/docs/tree/main/tools/viewfactor},
	Energy Systems Optimization Lab,
	University of Wisconsin–Madison.
	
\bibitem{Kerkhoff2021}
	Kerkhoff, J. A. and Wagner, M. J.,
	\textit{A Flexible Thermal Model for Solar Cavity Receivers Using Analytical View Factors},
	ASME Energy Sustainability Proceedings,
	July 22, 2021,
	\url{https://doi.org/10.1115/ES2021-63810}.
	
\bibitem{Bielefeld2024}
	Bielefeld, N. M.,
	\textit{RayTraceHeatTransfer.jl: Radiative Heat Transfer with Exchange Factors in Participating Media and Enclosures},
	version 0.22.1 (2026), first presented 2024, Zenodo,
	\url{https://doi.org/10.5281/zenodo.23355331},
	source code: \url{https://github.com/NikoBiele/RayTraceHeatTransfer.jl}.
	
\bibitem{Bielefeld2025}
	Bielefeld, N. M.,
	\textit{ConvolutionInterpolations.jl: High-Order Convolution Interpolation},
	version 0.28.2, (2026), first presented 2025, Zenodo,
	\url{https://doi.org/10.5281/zenodo.23354656},
	source code: \url{https://github.com/NikoBiele/ConvolutionInterpolations.jl}.
	
\bibitem{DanischKrumbiegel2021}
	Danisch, S. and Krumbiegel, J.,
	\textit{Makie.jl: Flexible high-performance data visualization for Julia},
	Journal of Open Source Software,
	The Open Journal,
	Vol. 6, no. 65, pp. 3349, 2021,
	\url{https://doi.org/10.21105/joss.03349}.

\bibitem{Montoison2023}
	Montoison, A. and Orban, D.,
	\textit{Krylov.jl: A Julia basket of hand-picked Krylov methods},
	Journal of Open Source Software,
	Vol. 8, no. 89, pp. 5187, 2023,
	\url{https://doi.org/10.21105/joss.05187}.
	
\bibitem{SymPy2017}
	Meurer, A. and Smith, C. P. and Paprocki, M. and \v{C}ert\'{i}k, O. and Kirpichev, S. B. and Rocklin, M. and Kumar, A. and Ivanov, S. and Moore, J. K. and Singh, S. and Rathnayake, T. and Vig, S. and Granger, B. E. and Muller, R. P. and Bonazzi, F. and Gupta, H. and Vats, S. and Johansson, F. and Pedregosa, F. and Curry, M. J. and Terrel, A. R. and Rou\v{c}ka, \v{S}. and Saboo, A. and Fernando, I. and Kulal, S. and Cimrman, R. and Scopatz, A.,
	\textit{SymPy: symbolic computing in Python},
	PeerJ Computer Science,
	3:e103,
	2017,
	\url{https://doi.org/10.7717/peerj-cs.103}.

\bibitem{deMoura2021}
	de Moura, L. and Ullrich, S.,
	\textit{The Lean 4 Theorem Prover and Programming Language},
	in: Platzer, A. and Sutcliffe, G. (Eds.), Automated Deduction -- CADE 28,
	Lecture Notes in Artificial Intelligence, Vol. 12699, Springer, Cham,
	pp. 625-635, 2021,
	\url{https://doi.org/10.1007/978-3-030-79876-5_37}.

\bibitem{Mathlib2020}
	The mathlib Community,
	\textit{The Lean Mathematical Library},
	in: Proceedings of the 9th ACM SIGPLAN International Conference on Certified Programs and Proofs (CPP 2020),
	pp. 367-381, 2020,
	\url{https://doi.org/10.1145/3372885.3373824}.

\end{thebibliography}
\end{document}